%% file: main.tex
\documentclass{article}

\usepackage{arxiv}

\usepackage[utf8]{inputenc}             
\usepackage[T1]{fontenc}                
\usepackage{hyperref}                   
\usepackage{url}                        
\usepackage{booktabs}                   
\usepackage{amsfonts}                   
\usepackage{nicefrac}                   
\usepackage{microtype}                  
\usepackage[bottom]{footmisc}
\usepackage{graphicx}
\usepackage{subfigure}
\usepackage{booktabs}                   
\usepackage{amsmath,amssymb,bm}
\usepackage{xcolor,xspace}
\usepackage{hyperref}
\usepackage{amsthm}
\usepackage{verbatim}
\usepackage{algorithm}
\usepackage{algorithmic}
\usepackage{comment}
\usepackage[normalem]{ulem}

\def\~#1{\mathbb{#1}}
\def\*#1{\mathbf{#1}}
\def\@#1{\mathcal{#1}}

\def\oneshotnp{{\texttt{NonParametricFilter}}\xspace}
\def\oneshot{{\texttt{DeterministicFilter}}\xspace}
\def\okm{{\texttt{ParametricFilter}}\xspace}
\def\off{{\texttt{Offline}}\xspace}
\def\tp{{\texttt{TwoPass}}\xspace}

\def\uni{{\texttt{Uniform}}}

\def\coreset{\*C}

\newtheorem{assumption}{Assumption}[section]
\newtheorem{definition}{Definition}[section]
\newtheorem{theorem}{Theorem}[section]
\newtheorem{corollary}{Corollary}[section]
\newtheorem{lemma}{Lemma}[section]

\title{Online Coresets for Clustering with Bregman Divergences}

\author{
  Rachit Chhaya\\
  IIT Gandhinagar\\
  \texttt{rachit.chhaya@iitgn.ac.in} \\
   \And
  Jayesh Choudhari\\
  University of Warwick\\
  \texttt{choudhari.jayesh@alumni.iitgn.ac.in} \\
   \And
 Anirban Dasgupta \\
  IIT Gandhinagar\\
  \texttt{anirbandg@iitgn.ac.in} \\
  \AND
   Supratim Shit \thanks{Corresponding author} \\
  IIT Gandhinagar \\
  \texttt{supratim.shit@iitgn.ac.in} \\
}

\begin{document}
\maketitle

\begin{abstract}
We present algorithms that create coresets in an online setting for clustering problems according to a wide subset of Bregman divergences. Notably, our coresets have a small additive error, similar in magnitude to the lightweight coresets~\cite{bachem2018scalable}, and take update time $O(d)$ for every incoming point where $d$ is dimension of the point. Our first algorithm gives online coresets of size $\tilde{O}(\mbox{poly}(k,d,\epsilon,\mu))$ for $k$-clusterings according to any $\mu$-similar Bregman divergence. We further extend this algorithm to show existence of a non-parametric coresets, where the coreset size is independent of $k$, the number of clusters, for the same subclass of Bregman divergences. Our non-parametric coresets are larger by a factor of $O(\log n)$ ($n$ is number of points) and have similar (small) additive guarantee. At the same time our coresets also function as lightweight coresets for non-parametric versions of the Bregman clustering  like DP-Means. While these coresets provide additive error guarantees, they are also significantly smaller (scaling with $O(\log n)$ as opposed to $O(d^d)$ for points in $\~R^d$) than the (relative-error) coresets obtained in~\cite{bachem2015coresets} for DP-Means. While our non-parametric coresets are existential, we give an algorithmic version under certain assumptions. 
\end{abstract}
 

\allowdisplaybreaks

\input{sections/introduction}
\input{sections/preliminary}
\input{sections/related}
\input{sections/parametric}
\input{sections/nonparametric}
\input{sections/experiment}
\input{sections/ack}

\bibliographystyle{unsrt}
\bibliography{reference}


\end{document}

%% file: sections/introduction.tex
\section{Introduction}
Clustering is perhaps one of the most frequently used operations in data processing,  and a canonical definition of the clustering problem is via the $k$-median, in which propose $k$ possible centers such that the sum of distances of every point to its closest center is minimized. There has been a plethora of work, both theoretical and practical, devoted to finding efficient and provable clustering algorithms in this $k$-median setting. However, most of this literature is devoted towards dissimilarity measures that are algorithmically easier to handle, namely the various $\ell_p$ norms, especially Euclidean.  
However, other dissimilarity measures, e.g. Kullback Leibler or Itakuro-Saito, are often more appropriate based on the data.  A mathematically elegant family of dissimilarity measures that have found wide use are the {\em Bregman divergences}, which include, for instance, the squared Euclidean distance, the Mahalanobis distance, Kullbeck-Leibler divergence, Itakuro-Saito dissimilarity and many others. 

While being mathematically satisfying, the chief drawback of working with Bregman divergences is algorithmic--- most of these divergences do not satisfy either symmetry or triangle inequality conditions. 
Hence, developing efficient clustering algorithms for these has been a much harder problem to tackle. Banerjee et. al. \cite{banerjee2005clustering} has done systematic study of the $k$-median clustering problem under Bregman divergence, and proposed algorithms that are generalization of the Lloyd's iterative algorithm for the Euclidean $k$-means problem.   
However, scalability remains a major issue.
Given that there are no theoretical bounds on the convergence of the Lloyd's algorithm in the general Bregman setting, a decent solution is often achieved only via running enough iterations as well as searching over multiple initializations. 
This is clearly expensive when the number of data points and the data dimension is large.   

Coresets, a data summarization technique to enable efficient optimization, has found multiple uses in many problems, especially in computational geometry and more recently in machine learning via randomized numerical linear algebraic techniques. The aim is to judiciously select (and reweigh) a set of points from the input points, so that solving the optimization problem on the coreset gives a guaranteed approximation to the optimization problem on the full data. 

In this work we explore the use of coresets to make Bregman clustering more efficient. Our aim is to give coresets that are small, the dependence on the number of points as well as the dimension should be {\em linear or better}-- it is not apriori clear that this can be achieved for all Bregman divergences. Most coresets for $k$-means and various linear algebraic problems, while being sublinear in the number of points, are often super-linear in the dimension of the data. However, in big-data setups, it is fairly common to have the number of dimensions to be almost of the same order of magnitude as the number of points. Coresets that trade-off being sublinear in the number of points while increasing the dependence on the dimension (to, say, exponential) might not be desirable in such scenarios. 

A further complication is the dependence of the coreset size on the number of clusters, as $k$, the number of clusters can be large, and more importantly, it can be unknown, to be determined only after exploratory analysis with clustering. When the number of clusters is unknown, it is unclear how to apply a coreset construction that needs knowledge of $k$. Recent work by Huang et. al. \cite{huang2020coresets} shows that for relative error coresets for Euclidean $k$-means, a linear dependence of coreset size on $k$ is both sufficient and inevitable.   

In this work, we tackle these questions for Bregman divergences. We develop coresets with {\em small additive error guarantees}. Such results have been obtained in the Euclidean setting by Bachem et. al. \cite{bachem2018scalable}, and in the online subspace embedding setting by Cohen et. al. \cite{cohen2016online}. 
We next show the existence of {\em non-parametric} coresets, where the coreset size is independent of $k$, the parameter representing number of cluster centers.  We utilize the {\em sensitivity} framework of~\cite{feldman2011unified} jointly with the barrier functions method of~\cite{batson2012twice} in order to achieve this.  A non-parametric coreset will be useful in problems such as DP-Means clustering \cite{bachem2015coresets} and extreme clustering \cite{kobren2017hierarchical}. We now formally describe the setup and list our contributions.

Given $\*A \in \~R^{n \times d}$ where
the rows (aka points) arrive in streaming fashion, let $\*A_{i} \in \~R^{i \times d}$ represent the first $i$ points that have arrived and $\*C_{i}$ be the coreset maintained for $\*A_{i}$. Let $\varphi_{i}$ denote the mean point of $\*A_{i}$, i.e., $\varphi_{i} = (1/i)\sum_{j \leq i} \*a_{j}$ and $\varphi$ is the mean of $\*A$. 
Let $\*X \in \~R^{k \times d}$ denote a candidate set of $k$ centers in $\~R^{d}$ and lets $f_{\*X}(\*A_i)$ be the total sum of distances of each point $\*a\in \*A_i$ from its closest center in $\*X$, according to a chosen Bregman divergence. We give algorithms which return coreset $\*C$ that ensures the following for any $\*X$, 
\begin{equation}{\label{eq:lwkm}}
 |f_{\*X}(\*C) - f_{\*X}(\*A)| \leq \epsilon (f_{\*X}(\*A) + f_{\varphi}(\*A))
\end{equation}

The following are our main contributions.
\begin{itemize}
 \item We give an algorithm named \okm (\textbf{Algorithm (\ref{alg:lwkm})}) which ensures property \eqref{eq:lwkm} for any $\*X \in \~R^{k \times d}$ with at least $0.99$ probability. \okm returns a coreset $\*C$ for $\*A$. It takes $O(d)$ update time and uses $O(d)$ working space. The expected size of the coreset $\*C$ is $O\Big(\frac{dk\log (1/\epsilon)}{\epsilon^{2}\mu^{2}}\big(\log n + \log \big(f_{\varphi}(\*A)\big) - \log \big(f_{\varphi_{2}}(\*a_{2})\big)\big)\Big)$ (Theorem \ref{thm:lwkm}). Here, $f_{\*x}(\cdot)$ is a $\mu$-similar Bregman divergence to some squared Mahalanobis distance.
 \item For the special case of $k$-means clustering, \okm~builds online coreset $\*C$ for $\*A$ which ensures property \eqref{eq:lwkm} for all $\*X \in \~R^{k \times d}$ with at least $0.99$ probability. The update time and working space are $O(d)$. The expected coreset size is $O\Big(\frac{dk\log (1/\epsilon)}{\epsilon^{2}}\big(\log n + \log \big(f_{\varphi}(\*A)\big) - \log \big(f_{\varphi_{2}}(\*a_{2})\big)\big)\Big)$ (Corollary \ref{cor:lwk}).
 \item We show that it is impossible to get a non-parametric coreset for clustering problem which ensures a relative error approximation to the optimal cost (Theorem \ref{thm:impossible}). We then give an existential result for non-parametric coreset with small additive error approximation. For this we present a method \oneshot (\ref{alg:lwdkm}) which uses an oracle while taking the sampling decision. The coreset ensures \eqref{eq:lwkm} for any $\*X$ with at most $n$ centers in $\~R^{d}$. Hence we call non-parametric coreset. The algorithm returns coreset of size $O\Big(\frac{\log n}{\epsilon^{2}\mu^{2}}\big(\log n + \log \big(f_{\varphi}(\*A)\big) - \log \big(f_{\varphi_{2}}(\*a_{2})\big)\big)\Big)$ (Theorem \ref{thm:lwdkm}). Here $f_{\*X}(\cdot)$ is a $\mu$-similar Bregman divergence to some squared Mahalanobis distance. 
 \item Again for special case of clustering based on squared euclidean distance, \oneshot~builds coreset $\*C$ for $\*A$ which ensures \eqref{eq:lwkm} for any $\*X$ with at most $n$ centers in $\~R^{d}$. The method returns a coreset of size $O\Big(\frac{\log n}{\epsilon^{2}}\big(\log n + \log \big(f_{\varphi}(\*A)\big) - \log \big(f_{\varphi_{2}}(\*a_{2})\big)\big)\Big)$ (Corollary \ref{cor:nplwk}).
 \item The coresets from \oneshot can be considered as non-parametric coreset for DP-Means clustering (Theorem \ref{thm:DPMeans}). The coreset size is $O\Big(\frac{\log n}{\mu^{2}\epsilon^{2}}\big(\log n + \log \big(f_{\varphi}(\*A)\big) - \log \big(f_{\varphi_{2}}(\*a_{2})\big)\big)\Big)$ opposed to $O(d^dk^{*}\epsilon^{-2})$ \cite{ackermann2009coresets}, where $k^{*}$ is the optimal centers for DP-Means clustering.
\item Under certain assumptions, we propose an algorithm named \oneshotnp (\textbf{Algorithm (\ref{alg:lwnpkm})}) which creates a coreset for non-parametric clustering.
\end{itemize}
Except for the existential result, the above contributions can also be made true in the online setting, i.e., at every point $i \in [n]$ the set $\*C_{i}$ maintained for $\*A_{i}$ ensures the guarantee with some constant probability by taking a union bound over all $i \in [n]$. Note that this is a stronger guarantee and in this case the expected sample size gets multiplied by a factor of $O(\log n)$. 

%% file: sections/preliminary.tex
\section{Preliminaries and Notations}
 Here we define the notation that we use in rest of the paper. First $n$ natural number set is represented by $[n]$. A bold lower case letter denotes a vector or a point for e.g. $\*a$, and a bold upper case letter denotes a matrix or set of points as defined by the context for e.g. $\*A$. In general $\*A$ has $n$ points each in $\~R^{d}$. $\*a_{i}$ denotes the $i^{th}$ row of matrix $\*A$ and $\*a^{j}$ denotes its $j^{th}$ column. We use the notation $\*A_{i}$ to denote the matrix or a set, formed by the first $i$ rows or points of $\*A$ seen till a time in the streaming setting. 
%
 \begin{definition}
  \textbf{Bregman divergence:} For any strictly convex, differentiable function $\Phi: \mathcal{Z} \to \~R$, the Bregman divergence with respect to $\Phi$, $\forall \*x, \*y \in \mathcal{Z}$ is,
  \begin{equation*}
   d_{\Phi}(\*y, \*x) = \Phi(\*y) - \Phi(\*x) - \nabla \Phi(\*x)^{T}(\*y - \*x)
  \end{equation*}
 \end{definition}
 
 We also denote $f_{\*x}(\*y) = d_{\Phi}(\*y, \*x)$. Throughout the paper for some set of centers $\*X$ in $\~R^{d}$ and point $\*a \in \~R^{d}$ we consider $f_{\*X}(\*a)$ as a cost function based on Bregman divergence. We define it as $f_{\*X}(\*a) = \min_{\*x \in \*X}f_{\*x}(\*a) = \min_{\*x \in \*X}d_{\Phi}(\*a,\*x)$, where $d_{\Phi}(\cdot)$ is some Bregman divergence as defined above. If the set of points $\*A$ have weights $\{w_a\}$ then we define $f_{\*x}(\*a) = w_{\*a}d_{\Phi}(\*a, \*x)$. 
 
 The Bregman divergence $d_\Phi$ is said to be  {\em $\mu$-similar} if it satisfies the following property-- $\exists \*M \succ 0$ such that, if $d_{\*M}(\*y, \*x) =  (\*y - \*x)^{T}\*M(\*y - \*x)$ denotes the squared Mahalanobis distance measure for $\*M$, then for all $\*x, \*y$,
  $\mu d_{\*M}(\*y, \*x) \leq d_{\Phi}(\*y, \*x) \leq d_{\* M} (\*y, \*x)$.
%

 Going forward, we also denote $f_{\*x}^{\*M}(\*a) = d_{\*M}(\*a, \*x)$, and hence, we have $\mu f_{\*x}^{\*M}(\*a) \leq f_{\*x}(\*a) \leq f_{\*x}^{\*M}(\*a)$, $\forall \*x$ and $\forall \*a \in \*A$. Due to this we say $f_{\*x}(\cdot)$ and $f_{\*x}^{\*M}(\cdot)$ are $\mu$ similar.
 For Euclidean $k$-means clustering $\*M$ is just an identity matrix and $\mu = 1$. It is known that a large set of Bregman divergences is $\mu$-similar, including KL-divergence, Itakura-Saito, Relative Entropy, Harmonic etc~\cite{ackermann2009coresets}. In Table~\ref{tab:bd}, we list the most common $\mu$-similar Bregman divergences, their corresponding $\*M$ and the $\*\mu$. In each case the $\lambda$ and $\nu$ refer to the minimum and maximum values of all coordinates over all points, i.e. the input is a subset of $[\lambda, \nu]^d$.
 
 \begin{table}[ht]
\caption{$\mu$-similar Bregman divergences} \label{tab:muBS}
\begin{center}
\begin{tabular}{lll}
 \textbf{Divergence} & $\mu$ & $\*M$ \\
\hline 
Squared-Euclidean & $1$ & $\*I_{d}$ \\
Mahalanobis$_N$ & $1$ & $\*N$ \\
Exponential-Loss & $e^{-(\nu - \lambda)}$ & $\frac{e^{\nu}}{2}\*I_{d}$ \\
Kullback-Leibler & $\frac{\lambda}{\nu}$ & $\frac{1}{2\lambda}\*I_{d}$ \\
Itakura-Saito & $\frac{\lambda^{2}}{\nu^{2}}$ & $\frac{1}{2\lambda^{2}}\*I_{d}$ \\
Harmonic$_\alpha$ $(\alpha > 0)$ & $\frac{\lambda^{\alpha+2}}{\nu^{\alpha+2}}$ & $\frac{\alpha(1-\alpha)}{2\lambda^{\alpha + 2}}\*I_{d}$ \\
Norm-Like$_\alpha$ $(\alpha > 2)$ & $\frac{\lambda^{\alpha-2}}{\nu^{\alpha-2}}$ & $\frac{\alpha(1-\alpha)}{2}\nu^{\alpha - 2}\*I_{d}$ \\
Hellinger-Loss & $2(1-\nu^{2})^{3/2}$ & ${2(1-\nu)^{-3/2}}\*I_{d}$ \\
\hline
\end{tabular}
\end{center}
\label{tab:bd}
\end{table}
 
 For the Bregman divergence clustering problem, the set $\*X$, called the query set, will represent the set of all possible candidate centers. There are two types of clustering, hard and soft clustering for Bregman divergence \cite{banerjee2005clustering}. In this work, by the term clustering, we refer only to the hard clustering problem. 
\paragraph{Coresets:}
A coreset acts as a small proxy for the original data in the sense that it can be used in place of the original data for a given optimization problem in order to obtain a provably accurate  approximate solution to the problem. Formally, for a non-negative cost function $f_{\*X}(\*a)$ with query $\*X$ and data point $\*a \in \*A$, a set of subsampled and appropriately reweighted points $\*C$ is coreset if $\forall \*X$,  $|\sum_{\*a \in \*A}f_{\*X}(\*a) - \sum_{\tilde{\*a} \in \*C}f_{\*X}(\tilde{\*a})| \leq \epsilon\sum_{\*a \in \*A}f_{\*X}(\*a)$ for some $\epsilon > 0$.

While coresets are typically defined for relative errors, additive error coresets can also be defined similarly. For  $\epsilon, \gamma > 0$, $\*C$ is an additive $(\epsilon, \gamma)$ coreset of $\*A$ if  
$\*C$ contains reweighted points from $\*A$, and $\forall \*X$, $|\sum_{\*a \in \*A}f_{\*X}(\*a) - \sum_{\tilde{\*a} \in C}f_{\*X}(\tilde{\*a})| \leq \epsilon\sum_{\*a \in \*A}f_{\*X}(\*a) + \gamma$. 
The coresets that are presented here satisfies such additive guarantees.

For a dataset $\*A$, a query space $\@X$ that denotes candidate solutions to an optimization problem, and a cost function $f_{\*X}(\cdot)$, \cite{langberg2010universal} define sensitivity scores that capture the relative importance of each point for the problem and can be used to construct a probability distribution. 
The coreset is then created by sampling points according to this distribution.
The sensitivity of a point $\*a$ is defined as $s_{\*a} = \sup_{\*X \in \@X} \frac{ f_{\*X}(\*a)}{\sum_{\*a' \in \*A} f_{\*X}(\*a')}$.

\textbf{Lightweight Coresets:} Lightweight coresets were introduced by \cite{bachem2018scalable} for clustering based on $\mu$-similar Bregman divergence. These coresets give an additive error guarantee and they were built based on sensitivity framework. For a dataset $\*A \in \~R^{n \times d}$, and a cost function $f_{\*X}(\cdot)$ for some $\*X \in \~R^{k \times d}$ the sensitivity of a point $\*a$ is $s_{\*a} = \sup_{\*X} \frac{f_{\*X}(\*a)}{f_{\*X}(\*A) + f_{\varphi}(\*A)}$. Here $\varphi$ is the mean point of the entire dataset $\*A$ i.e., $\varphi = \sum_{i \leq n}\*a_{i}/n$ and $f_{\varphi}(\*A) = \sum_{\*a \in \*A}f_{\varphi}(\*a)$. 

In this work, we define $\*C$ to be an $(\epsilon, \gamma)$-additive error {\em non-parametric coreset} if its size is independent of $k$ (number of centres) and ensures $|f_{\*X}(\*C) - f_{\*X}(\*A)| \leq \epsilon f_{\*X}(\*A) + \gamma$ for any query $\*X \in \~R^{k \times d}$ for all integers $k \in [n]$. 

We use the following theorems in this paper.
\begin{theorem}{\label{thm:bernstein}}
\textbf{Bernstein's inequality \cite{dubhashi2009concentration}} Let the scalar random variables $x_{1}, x_{2}, _{\cdots}, x_{n}$ be independent that satisfy $\forall i \in [n]$,  
$\vert x_{i}-\~E[x_{i}]\vert \leq b$. 
Let $X = \sum_{i} x_{i}$ and let $\sigma^{2} = \sum_{i} \sigma_{i}^{2}$ be the variance of $X$. 
Then for any $t>0$,
\begin{center}
 $\mbox{Pr}\big(X > \~E[X] + t\big) \leq \exp\bigg(\frac{-t^{2}}{2\sigma^{2}+bt/3} \bigg)$
\end{center}
\end{theorem}
%
%
\begin{theorem}{\label{thm:strongCorest}}
 \textbf{\cite{chhaya2020coresets}} Let $\*A$ be the dataset, $\*X$ be the query space of dimension $D$, and for $\*x \in \*X$,  let $f_{\*x}(\cdot)$ be the cost function. Let $s_j$ be the sensitivity of the $j^{th}$ row of $\*A$, and the sum of sensitivities be $S$. Let $(\epsilon,\delta) \in (0,1)$. Let $r$ be such that
\begin{equation*}
r \geq O{\Big(\frac{S}{\epsilon^2}(D\log{\frac{1}{\epsilon}} + \log{\frac{1}{\delta}})\Big)}
\end{equation*}
$\*C$ be a matrix of $r$ rows, each 
sampled i.i.d from $\*A$ such that each $\tilde{\*a}_i \in \*C$ is chosen to be $\*a_j$, with weight $\frac{S}{rs_j}$, with probability $\frac{s_j}{S}$, for $j \in [n]$. Then $\*C$ is an $\epsilon$-coreset of $\*A$ for function $f()$, with probability at least $1 - \delta$.
\end{theorem}

We use the above Theorem to bound the coreset size. Note that the Theorem considers a multinomial sample where a point $\tilde{\*a}_{i}$ in coreset $\*C$ is $\*a_j$ and weight  $\frac{S}{rs_j}$ for $j \in [n]$ with probability $\frac{s_{j}}{S}$. Instead in our approach we get $\tilde{\*a}_{i}$ as $\*a_{i}$, with weight $1/\min\{1,rs_i\}$, with probability $\min\{rs_{i},1\}$ or it is $\emptyset$, with weight $0$, with probability $1-\min\{rs_{i},1\}$. However, the same Theorem as above applies. 

%% file: sections/related.tex
\section{Related Work}
The term coreset was first introduced in \cite{agarwal2004approximating}  and there has been a significant amount of work on coresets since then. 
Interested readers can look at \cite{bachem2017practical, woodruff2014sketching} and the references therein. Using sensitivities to construct coresets was introduced in \cite{langberg2010universal} and further generalized by \cite{feldman2011unified}.
Coresets for clustering problems such as $k$-means clustering have been extensively studied \cite{bachem2018scalable, har2004coresets, cohen2015dimensionality, braverman2016new, barger2016k, feldman2016dimensionality, bachem2018one, feldman2020turning}. 
In \cite{cohen2015dimensionality} the authors reduce the k-means problem to a constrained low rank approximation problem. 
They show that a constant factor approximation can be achieved by just $O(\epsilon^{-2}\log k)$ size coreset and for $(1 \pm \epsilon)$ relative error approximation they give coreset of size $O(k\epsilon^{-2})$. 
In \cite{barger2016k, feldman2016dimensionality}, the authors discuss a deterministic algorithm for creating coresets for clustering problem which ensure a relative error approximation. The streaming version of \cite{barger2016k} returns a coreset of size $O(k^{\epsilon^{-2}}\epsilon^{-2}\log n)$ which ensures a $(1\pm \epsilon\log n)$ relative error approximation. Feldman et. al. \cite{feldman2016dimensionality} reduce the problem of $k$-means clustering to $\ell_{2}$ frequent item approximation. 
The streaming version of the algorithm returns a coreset size of $O(k^{2}\epsilon^{-2}\log^{2} n)$. 
In \cite{bachem2018one} authors give an algorithm which returns a one shot coreset for all $p$ euclidean distance k-clustering problem, where $p \in [1,p_{\max}]$. 
Their algorithm creates a grid over the range $[1,p_{\max}]$ and based on the sensitivity at each grid point the coreset is built. 
It returns a coreset of size $\tilde{O}(16^{p_{\max}}dk^{2})$ for which it takes $\tilde{O}(ndk)$ ensuring $(1\pm \epsilon)$ relative error approximation. 
In a slightly different line \cite{boutsidis2013deterministic} gives a deterministic algorithm for feature selection in k-means problem. 
In \cite{bachem2018scalable}, the authors give an algorithm to create a create coreset which only takes $O(nd)$ time and returns a coreset of size $O(dk\epsilon^{-2}\log k)$ at a cost of small additive error approximation. 
Their algorithm can further be extended for clustering based on Bregman divergences which are $\mu$-similar to squared Mahalanobis distance. In \cite{lucic2016strong} the authors give  algorithms to create such coresets for both hard and soft clustering based on $\mu$-similar Bregman Divergence. 


There are several online algorithms for k-means clustering \cite{liberty2016algorithm, lattanzi2017consistent, bhaskara2020robust}. In \cite{liberty2016algorithm}, the authors give an online algorithm that maintains a set of centers such that k-means cost on these centres is $\tilde{O}(W^{*})$ where $W^{*}$ is the optimal k-means cost. \cite{lattanzi2017consistent} improves this result and gives a robust algorithm which can also handle outliers in the dataset.

For our analysis we use theorem 3.2 in \cite{chhaya2020coresets}, where authors show that the coreset built using sensitivity framework has a sampling complexity that only depends on $O(S)$ instead of $O(S \log(S))$ as in \cite{braverman2016new} but with an additional factor of $\log(1/\epsilon)$. Due to this, our coreset size for clustering based on $\mu$-similar Bregman divergence only has dependence of $O(1/\mu)$, unlike in \cite{bachem2018scalable, lucic2016strong} where the dependence is $O(1/ \mu^2)$.

%% file: sections/parametric.tex
\section{Online Coresets for Clustering}
Here we state our first algorithm \okm which creates a coreset in an online manner for clustering based on Bregman divergence, i.e., for the $i^{th}$ incoming point we take the sampling decision without looking at the $(i+1)^{th}$ point.
The algorithm starts with knowledge of the Bregman divergence $d_\Phi()$. 
It is important to note that for a fixed $d_\Phi$, as $\*A$ changes, both $\*M$ and $\mu$ also change \cite{ackermann2009coresets, lucic2016strong} (table \ref{tab:bd}). 
Fortunately, updating the Mahalanobis matrix requires maintaining only two simple statistic of the data. 
On arrival of the input point $\*a_i$, the algorithm first updates both the Mahalanobis matrix $\*M_i$ as well as $\mu_i$ and then uses it to compute the upper bound for the sensitivity score. 
This score is then used to decide whether $\*a_i$ should be stored in the coreset. If selected, the point $\*a_i$ is stored with an appropriate weight $\omega_i$.
\begin{algorithm}[htpb]
  \caption{\okm}{\label{alg:lwkm}}
  \begin{algorithmic}
  \REQUIRE Streaming points $\*a_{i}, i = 1, 2, \ldots, n; r > 0$
  \ENSURE $(\text{Coreset } \coreset, \text{Weights } \Omega)$
  \STATE $\coreset_{0} = \Omega_{0} = \varphi_{0} = \emptyset; S = 0$
  \STATE $\lambda = \|\*a_{1}\|_{\min}; \quad \nu = \|\*a_{1}\|_{\max}$
  \WHILE {$i \leq n$}
    \STATE $\lambda = \min\{\lambda,\|\*a_{i}\|_{\min}\}; \nu = \max\{\nu,\|\*a_{i}\|_{\max}\}$
    \STATE Update $\*M_{i};\ \mu_{i} = \lambda/\nu$ 
    \STATE $\varphi_{i} = ((i-1)\varphi_{i-1} + \*a_{i})/i ; S = S + f_{\varphi_{i}}^{\*M_{i}}(\*a_{i})$
    \IF{$i = 1$}
      \STATE $p_{i} = 1$
    \ELSE
      \STATE $l_{i} = \frac{2f_{\varphi_{i}}^{\*M_{i}}(\*a_{i})}{\mu_{i} S} + \frac{8}{\mu_{i}(i-1)} ; p_{i} = \min\{1,rl_{i}\}$
    \ENDIF
    \STATE Set $\*c_{i}$ and $\omega_{i}$ as \\
    $\begin{cases}\*a_{i} \mbox{ and } 1/p_{i} \qquad \qquad \mbox{w. p. } p_{i} \\  
    \emptyset \mbox{ and } 0 \qquad \qquad \qquad \mbox{else} \end{cases}$
    \STATE $(\*C_{i}, \Omega_{i}) = (\*C_{i-1}, \Omega_{i-1}) \cup (\*c_{i},\omega_{i})$
  \ENDWHILE
  \STATE Return $(\*C, \Omega)$
  \end{algorithmic}
\end{algorithm}

Notice that as working space the algorithm only needs to maintain the current mean $\varphi_i$, the diagonal matrix $\*M_i$, and the values $S,\lambda$, and $\nu$. For the case when $d_\Phi$ is Mahalanobis distance, we consider that $\*M$ and $\mu$ are know a priori and the algorithm uses them for points $\*a_{i}$. Hence, in the case of Mahalanobis distance, the update time and the working space are both $O(d^{2})$. For all other divergences in Table~\ref{tab:bd} however, the matrix $\*M$ is diagonal and hence both the update time and the working space are $O(d)$ only. 

Let $\*A_{i}$ be the dataset formed by first $i$ data points. 
Algorithm (\okm \ref{alg:lwkm}) updates $\*M_{i}$ and $\mu_{i}$ for $\*A_{i}$. By a careful analysis, we show that even when these are updated online, we achieve a one pass online algorithm that creates an additive error coreset.

Before stating the main results, we now give some intuition why updating the Mahalanobis matrix works.
Notice that, for every incoming point, \okm maintains a positive definite matrix $\*M_{i}$, a range $[\lambda, \nu]$ and the mean of $\*A_{i}$ as $\varphi_{i}$. 
Here $\lambda$ is the smallest absolute value in $\*A_{i}$, i.e., $\lambda = \|\*A_{i}\|_{\min}$ and $\nu$ is the highest absolute value in $\*A_{i}$, i.e., $\nu = \|\*A_{i}\|_{\max}$. With this $\lambda$ and $\nu$ the algorithm computes $\*M_i$ and 
$\mu_{i}$ as per Table \ref{tab:muBS}. 
Hence we have, $\mu_{i}f_{\*X}^{\*M_{i}}(\*a_{j}) \leq f_{\*X}(\*a_{j}) \leq f_{\*X}^{\*M_{i}}(\*a_{j})$, $\forall \*X$ and $\forall \*a_{j} \in \*A_{i}$. 

We note the following useful observation that is immediate, based on the formula for the matrix $\*M$  and the scalar $\mu$ in the Table \ref{tab:muBS}. The lemma applies for all Bregman divergence, but Mahalanobis distance and we use the lemma to show the algorithm's correctness.
\begin{lemma}{\label{lemma:obs}}
 For all Bregman divergences in Table \ref{tab:muBS}, for $j \le i$, $\mu_{j} \geq \mu_{i}$ and  $\*M_{j} \preceq \*M_{i}$.
\end{lemma} 
\begin{proof}{\label{proof:obs}}
 At any $i^{th}$ point we have $\lambda = \|\*A_{i}\|_{\min}$ and $\nu = \|\*A_{i}\|_{\max}$, i.e., the smallest and largest absolute values in $\*A_{i}$. Further we have $\|\*A_{j}\|_{\min} \geq \|\*A_{i}\|_{\min}$ and $\|\*A_{j}\|_{\max} \leq \|\*A_{i}\|_{\max}$ for $j \leq i$. By using the formula for $\*M$ for all Bregman divergences given in Table 1 we have $\*M_{j} \preceq \*M_{i}$ and $\mu_{j} \geq \mu_{i}$ to be always true for $j \leq i$. 
 \end{proof}
By a careful analysis in the following lemma, we show that the scores $l_i$ defined in \okm, upper bound the lightweight sensitivity scores of $\*a_i$ with respect to $\*A_{i-1}$, and that the sum of $l_i$'s is bounded. 
\begin{lemma}{\label{lemma:lwkScore}}
 For points coming in streaming manner $\forall i \in [n]$, the $l_{i}$  defined in \okm~, upper bounds the lightweight sensitivity score:
 \begin{eqnarray}{\label{eq:lwkmSensitivity}}
  \sup_{\*X \in \@X}\frac{f_{\*X}(\*a_{i})}{f_{\*X}(\*A_{i-1})+f_{\varphi_{i}}(\*A_{i})}
 \end{eqnarray}

 Furthermore, $\sum_{i \leq n}l_{i} \leq (8\log n + 4\log \big(f_{\varphi}^{\*M}(\*A)\big) - 4\log \big(f_{\varphi_{2}}^{\*M_{2}}(\*a_{2})\big)/\mu$.
\end{lemma}

\begin{proof}{\label{proof:lwkScore}}
 At $\*A_{i}$, let $(\mu_{i},\*M_{i})$ be such that $\mu_{i}f_{\*x}^{\*M_{i}}(\*a_{j}) \leq f_{\*x}(\*a_{j}) \leq f_{\*x}^{\*M_{i}}(\*a_{j})$ and $\varphi_{i} = \frac{\sum_{j \leq i}\*a_{j}}{i}$. For any $\*X \in \~R^{k \times d}$, each point $\*a_{j} \in \*A_{i-1}$ has some closest point $\*x_{l} \in \*X$. Hence for such pair $\{\*a_{j},\*x_{l}\}$, we have $f_{\*x_{l}}^{\*M_{i}}(\varphi_{i}) \leq 2f_{\*x_{l}}^{\*M_{i}}(\*a_{j}) + 2f_{\varphi_{i}}^{\*M_{i}}(\*a_{j})$. So $(i-1)f_{\*X}(\varphi_{i}) \leq 2\sum_{\*a_{j} \in \*A_{i-1}}(f_{\*X}(\*a_{j}) + f_{\varphi_{i}}(\*a_{j})) = 2f_{\*X}(\*A_{i-1}) + 2f_{\varphi_{i}}(\*A_{i-1})$. We use this triangle inequality in the following analysis, which holds $\forall \*X \in \~R^{k \times d}$,
 \begin{eqnarray*}
  \frac{f_{\*X}(\*a_{i})}{f_{\*X}(\*A_{i-1})+f_{\varphi_{i}}(\*A_{i})} &\stackrel{(i)}{\leq}& \frac{f_{\*X}^{\*M_{i}}(\*a_{i})}{f_{\*X}(\*A_{i-1})+f_{\varphi_{i}}(\*A_{i})} \\
  &\leq& \frac{2f_{\varphi_{i}}^{\*M_{i}}(\*a_{i}) + 2f_{\*X}^{\*M_{i}}(\varphi_{i})}{f_{\*X}(\*A_{i-1})+f_{\varphi_{i}}(\*A_{i})} \\
  &\leq& \frac{2f_{\varphi_{i}}^{\*M_{i}}(\*a_{i}) + \frac{4}{i-1}f_{\varphi_{i}}^{\*M_{i}}(\*A_{i-1}) + \frac{4}{i-1}f_{\*X}^{\*M_{i}}(\*A_{i-1})}{f_{\*X}(\*A_{i-1})+f_{\varphi_{i}}(\*A_{i})} \\
  &\stackrel{(ii)}{\leq}& \frac{2f_{\varphi_{i}}^{\*M_{i}}(\*a_{i}) + \frac{4}{i-1}f_{\varphi_{i}}^{\*M_{i}}(\*A_{i-1}) + \frac{4}{i-1}f_{\*X}^{\*M_{i}}(\*A_{i-1})}{\mu_{i}(f_{\*X}^{\*M_{i}}(\*A_{i-1})+f_{\varphi_{i}}^{\*M_{i}}(\*A_{i}))} \\
  &=& \frac{2f_{\varphi_{i}}^{\*M_{i}}(\*a_{i}) + \frac{4}{i-1}f_{\varphi_{i}}^{\*M_{i}}(\*A_{i-1})}{\mu_{i}(f_{\*X}^{\*M_{i}}(\*A_{i-1})+f_{\varphi_{i}}^{\*M_{i}}(\*A_{i}))} + \frac{\frac{4}{i-1}f_{\*X}^{\*M_{i}}(\*A_{i-1})}{\mu_{i}(f_{\*X}^{\*M_{i}}(\*A_{i-1})+f_{\varphi_{i}}^{\*M_{i}}(\*A_{i}))} \\
  &\leq& \frac{2f_{\varphi_{i}}^{\*M_{i}}(\*a_{i}) + \frac{4}{i-1}f_{\varphi_{i}}^{\*M_{i}}(\*A_{i-1})}{\mu_{i}(f_{\*X}^{\*M_{i}}(\*A_{i-1})+f_{\varphi_{i}}^{\*M_{i}}(\*A_{i}))} + \frac{4}{\mu_{i}(i-1)} \\
  &\leq& \frac{2f_{\varphi_{i}}^{\*M_{i}}(\*a_{i})}{\mu_{i}f_{\varphi_{i}}^{\*M_{i}}(\*A_{i})} + \frac{\frac{4}{i-1}f_{\varphi_{i}}^{\*M_{i}}(\*A_{i-1})}{\mu_{i}f_{\varphi_{i}}^{\*M_{i}}(\*A_{i})} + \frac{4}{\mu_{i}(i-1)} \\
  &\stackrel{(iii)}{\leq}& \frac{2f_{\varphi_{i}}^{\*M_{i}}(\*a_{i})}{\mu_{i}f_{\varphi_{i}}^{\*M_{i}}(\*A_{i})} + \frac{8}{\mu_{i}(i-1)} \\
  &\leq& \frac{2f_{\varphi_{i}}^{\*M_{i}}(\*a_{i})}{\mu_{i}\sum_{j\leq i}f_{\varphi_{j}}^{\*M_{j}}(\*a_{j})} + \frac{8}{\mu_{i}(i-1)}
 \end{eqnarray*}
 The inequality $(i)$ is due to $\mu_{i}$ similarity, i.e., $f_{\*X}(\*a_{i}) \leq f_{\*X}^{\*M_{i}}(\*a_{i})$. Next couple of inequalities are by applying triangle inequality on the numerator. In the $(ii)$ inequality we use the $\mu_{i}$ similarity lower bound on the denominator term. We reach to $(iii)$ inequality by upper bounding the second and third term by $4/(\mu_{i}(i-1))$. In the final inequality we use the lemma \ref{lemma:obs} from which we have $\*M_{j} \preceq \*M_{i}$ for $j \leq i$. Further by the property of Bregman divergence we know that $\varphi_{i-1} = \mbox{arg}\min_{\*x}f_{\*x}(\*A_{i-1})$, so we have $f_{\varphi_{i}}(\*A_{i}) = f_{\varphi_{i}}(\*A_{i-1}) + f_{\varphi_{i}}(\*a_{i})\geq f_{\varphi_{i-1}}(\*A_{i-1}) + f_{\varphi_{i}}(\*a_{i}) \geq \sum_{j \leq i}f_{\varphi_{j}}(\*a_{i})$. Hence we have $f_{\varphi_{i}}^{\*M_{i}}(\*A_{i}) \geq \sum_{j\leq i}f_{\varphi_{j}}^{\*M_{j}}(\*a_{j})$. 
 
 Next, in order to upper bound $\sum_{i \leq n}l_{i}$, consider the denominator term of $l_{i}$ as follows,
 \begin{eqnarray*}
  \sum_{j \leq i}f_{\varphi_{j}}^{\*M_{j}}(\*a_{j}) &=& \sum_{j \leq i-1}f_{\varphi_{j}}^{\*M_{j}}(\*a_{j}) + f_{\varphi_{i}}^{\*M_{i}}(\*a_{i}) \\
  &=& \sum_{j \leq i-1}f_{\varphi_{j}}^{\*M_{j}}(\*a_{j})\bigg(1 + \frac{f_{\varphi_{i}}^{\*M_{i}}(\*a_{i})}{\sum_{j \leq i-1}f_{\varphi_{j}}^{\*M_{j}}(\*a_{j})}\bigg) \\
  &\geq& \sum_{j \leq i-1}f_{\varphi_{j}}^{\*M_{j}}(\*a_{j})\bigg(1+\frac{f_{\varphi_{i}}^{\*M_{i}}(\*a_{i})}{\sum_{j \leq i}f_{\varphi_{j}}^{\*M_{j}}(\*a_{j})}\bigg) \\
  &=& \sum_{j \leq i-1}f_{\varphi_{j}}^{\*M_{j}}(\*a_{j})(1 + q_{i}) \\
  &\stackrel{(i)}{\geq}& \exp(q_{i}/2)\sum_{j \leq i-1}f_{\varphi_{j}}^{\*M_{j}}(\*a_{j}) \\
 \exp(q_{i}/2) &\leq& \frac{\sum_{j \leq i}f_{\varphi_{j}}^{\*M_{j}}(\*a_{j})}{\sum_{j \leq i-1}f_{\varphi_{j}}^{\*M_{j}}(\*a_{j})}
\end{eqnarray*}
where for inequality $(i)$ we used that $q_{i} = \frac{f_{\varphi_{i}}^{\*M_{i}}(\*a_{i})}{\sum_{j \leq i}f_{\varphi_{j}}^{\*M_{j}}(\*a_{j})} \leq 1$ and hence we have $(1 + q_{i}) \geq \exp(q_{i}/2)$. Now as we know that $\sum_{j \leq i}f_{\varphi_{j}}^{\*M_{j}}(\*a_{j}) \geq \sum_{j \leq i-1}f_{\varphi_{j}}^{\*M_{j}}(\*a_{j})$ hence the following product results into a telescopic product and we get,
\begin{eqnarray*}
 \prod_{2 \leq i \leq n} \exp(q_{i}/2) &\leq& \frac{\sum_{j \leq n}f_{\varphi_{j}}^{\*M_{j}}(\*a_{j})}{f_{\varphi_{2}}^{\*M_{2}}(\*a_{2})}
\end{eqnarray*}
So by taking logarithm of both sides we get $\sum_{2 \leq i \leq n} q_{i} \leq 2\log \big(f_{\varphi}^{\*M}(\*A)\big) - 2\log \big(f_{\varphi_{2}}^{\*M_{2}}(\*a_{2})\big)$. Further incorporating the terms $\frac{8}{\mu_{i}(i-1)}$ we have $l_{i} = \frac{2q_{i}}{\mu_{i}} + \frac{8}{\mu_{i}(i-1)}$. Hence, $\sum_{2 \leq i \leq n} l_{i} \leq 4\mu^{-1}(\log n + \log \big(f_{\varphi}^{\*M}(\*A)\big) - \log \big(f_{\varphi_{2}}^{\*M_{2}}(\*a_{2})\big))$. Where $\mu = \mu_{n} \leq \mu_{i}$ and $\*M \succeq \*M_{n} \succeq \*M_{i}$ for all $i \leq n$.
\end{proof}

Note that the upper bounds and the sum are independent of $k$, i.e., number of clusters one is expecting in the data. The following Lemma claims that by sampling enough points based on $l_{i}$ one can ensure the additive error coreset property with a high probability.
\begin{lemma}{\label{lemma:lwkGuarantee}}
 For clustering with $k$ centers in $\~R^{d}$, with $r = O\Big(\frac{dk\log (1/\epsilon)}{\epsilon^{2}}\Big)$ in \okm, the returned coreset $\*C$ satisfies the guarantee as in \eqref{eq:lwkm} at $i = n$, $\forall \*X \in \~R^{k \times d}$ with at least $0.99$ probability.
\end{lemma}
\begin{proof}{\label{proof:lwkGuarantee}}
For some fixed (query) $\*X \in \~R^{k \times d}$ consider the following random variable.
\[ w_{i} =
  \begin{cases}
    (1/p_{i} - 1)f_{\*X}(\*a_{i})  & \quad \text{with probability } p_{i} \\
    -f_{\*X}(\*a_{i}) & \quad \text{with probability } (1-p_{i})
  \end{cases}
\]
Note that $\~E[w_{i}] = 0$ and with $p = 1$ we get $|w_{i}| = 0$. The algorithm uses the sampling probability $p_{i} = \min\{rl_{i},1\}$. Now we bound the term $|w_{i}|$. In the case when $p_{i} < 1$ and $\*a_{i}$ is sampled we have,
\begin{eqnarray*}
 |w_{i}| &\leq& \frac{1}{p_{i}}f_{\*X}(\*a_{i}) \\
 &=& \frac{f_{\*X}(\*a_{i})}{rl_{i}} \\
 &\leq& \frac{(f_{\*X}(\*A_{i-1}) + f_{\varphi_{i}}(\*A_{i}))f_{\*X}(\*a_{i})}{rf_{\*X}(\*a_{i})} \\
 &\leq& \frac{(f_{\*X}(\*A) + f_{\varphi}(\*A))}{r}
\end{eqnarray*}
Here $\varphi = \varphi_n$ is the mean of the entire data $\*A$. Next if the point $\*a_{i}$ is not sampled then we know for sure that $p_{i} < 1$, hence, using the Lemma 4.2, we have that, 
\begin{eqnarray*}
 1 > rl_{i} &\geq& \frac{rf_{\*X}(\*a_{i})}{(f_{\*X}(\*A_{i-1}) + f_{\varphi_{i}}(\*A_{i}))} \\
 f_{\*X}(\*a_{i}) &\leq& \frac{(f_{\*X}(\*A) + f_{\varphi}(\*A))}{r}
\end{eqnarray*}
So we have $|w_{i}| \leq b = (f_{\*X}(\*A) + f_{\varphi}(\*A))/r$. Next we bound the $\mbox{var}(\sum_{i \leq n} w_{i}) = \sum_{i \leq n} \~E[w_{i}^{2}]$. Note that a single term $\~E[w_{i}^{2}]$ for $p_{i} < 1$ is,
\begin{eqnarray*}
 \~E[w_{i}^{2}] &=& \big(p_{i}(1/p_{i}-1)^{2} + (1-p_{i})\big)f_{\*X}(\*a_{i})^{2} \\
 &\leq& \frac{1}{p_{i}}f_{\*X}(\*a_{i})^{2} \\
 &=& \frac{f_{\*X}(\*a_{i})^{2}}{rl_{i}} \\
 &\leq& \frac{(f_{\*X}(\*A_{i-1}) + f_{\varphi_{i}}(\*A_{i}))f_{\*X}(\*a_{i})^{2}}{rf_{\*X}(\*a_{i})} \\
 &\leq& \frac{f_{\*X}(\*a_{i})(f_{\*X}(\*A) + f_{\varphi}(\*A))}{r}
\end{eqnarray*}
So we get,
\begin{eqnarray*}
 \mbox{var}\Big(\sum_{i \leq n} w_{i}\Big) &=& \sum_{i \leq n} \~E[w_{i}^{2}] \\
 &\leq& \sum_{i \leq n} \frac{f_{\*X}(\*a_{i})(f_{\*X}(\*A) + f_{\varphi}(\*A))}{r} \\
 &\leq& \frac{(f_{\*X}(\*A) + f_{\varphi}(\*A))^{2}}{r}
\end{eqnarray*}
 Now by applying Bernstein's inequality (\ref{thm:bernstein}) on $\sum_{i \leq n}w_{i}$ with $t = \epsilon(f_{\*X}(\*A) + f_{\varphi}(\*A))$ we bound the probability $\~P = \mbox{Pr}\Big(|(f_{\*X}(\*A) - f_{\*X}(\*C)| \geq \epsilon(f_{\*X}(\*A) + f_{\varphi}(\*A))\Big)$ as follows,
 \begin{eqnarray*}
  \~P &\leq& \exp\bigg(\frac{-\epsilon^{2}(f_{\*X}(\*A) + f_{\varphi}(\*A))^{2}}{\epsilon(f_{\*X}(\*A) + f_{\varphi}(\*A))^{2}/3r + 2(f_{\*X}(\*A) + f_{\varphi}(\*A))^{2}/r}\bigg) \\
  &=& \exp\bigg(\frac{-r\epsilon^{2}}{(\epsilon/3 + 2)}\bigg)
 \end{eqnarray*}
 So to get the above event with at least $0.99$ probability it is enough to set $r$ to be $\theta\Big(\frac{1}{\epsilon^{2}}\Big)$. Note that the above is guaranteed for a fixed $\*X \in \~R^{k\times d}$. 
 
 Now we show that coreset $\*C$ can be made strong coreset by taking a union bound over a set of queries. To ensure the guarantee in Lemma 4.3 for all $\*X \in \~R^{k\times d}$, we take a union bound over the $\epsilon/2$-net of $\~R^{k\times d}$ \cite{woodruff2014sketching, lucic2016strong}. Such a net will have at most $O(\epsilon^{-dk})$ queries. To ensure a strong a coreset guarantee it is enough to set $r$ as $\Theta\Big(\frac{dk\log (1/\epsilon)}{\epsilon^{2}}\Big)$.
\end{proof}

Utilizing Lemmas \ref{lemma:lwkScore} and \ref{lemma:lwkGuarantee}, where we take an union bound over the $\epsilon$-net of the query space, we get the following theorem. 
\begin{theorem}{\label{thm:lwkm}}
 For points coming in streaming fashion, \okm returns a coreset $\*C$ for the clustering based on Bregman divergence such that for all $\*X \in \~R^{k \times d}$, with at least $0.99$ probability $\*C$ ensures the guarantee \eqref{eq:lwkm}. Such a coreset has expected sample size of $O\Big(\frac{dk \log (1/\epsilon)}{\mu\epsilon^{2}}\big(\log n + \log \big(f_{\varphi}^{\*M}(\*A)\big) - \log \big(f_{\varphi_{2}}^{\*M_{2}}(\*a_{2})\big)\big)\Big)$. \okm~takes $O(d)$ update time and uses $O(d)$ as working space. 
\end{theorem}

The expected sample size of the coreset $\*C$ returned by \okm is bounded by $r\sum_{i \leq n}l_{i}$. Using Lemma \ref{lemma:lwkScore} and Lemma \ref{lemma:lwkGuarantee} we set the values of $r$ and $\sum_{i \leq n}l_{i}$, and obtain the expected sample size to be $O\Big(\frac{dk \log (1/\epsilon)}{\mu\epsilon^{2}}\big(\log n + \log \big(f_{\varphi}^{\*M}(\*A)\big) - \log \big(f_{\varphi_{2}}^{\*M_{2}}(\*a_{2})\big)\big)\Big)$. Further by using the $\mu$-similarity one can rewrite the expected sample size as $O\Big(\frac{dk \log (1/\epsilon)}{\mu^{2}\epsilon^{2}}\big(\log n + \log \big(f_{\varphi}(\*A)\big) - \log \big(f_{\varphi_{2}}(\*a_{2})\big)\big)\Big)$.

The algorithm requires a working space of $O(d)$ which is to maintain the mean(centre) $\varphi_{i}$, $\mu_{i}$ and $\*M_{i}$. Further for every incoming point \okm~only needs to compute the distance between the point and the current mean hence the running time of the entire algorithm is $O(nd)$, which is why it is easy to scale for large $n$. Note that although Theorem~\ref{thm:lwkm} gives the guarantees of \okm~at the last instance, but using the same analysis technique one can ensure an equivalent guarantee for any $i^{th}$ instance, but taking an union bound-- this requires a factor of $\log(n)$ for sample size, i.e. ensuring the guarantee \eqref{eq:lwkm} by $\*C_{i}$ for $\*A_{i}$, $\forall i \in [n]$.

Note that \okm~returns a smaller coreset $\*C$ compare to offline coresets \cite{bachem2018one, lucic2016strong} but at a cost of additive factor approximation that depends on the structure of the data. Further unlike \cite{bachem2018scalable, lucic2016strong} our sampling complexity only depends on $1/\mu$. \okm can be easily generalized to create coresets for weighted clustering where each point $\*a_{i}$ has some weight $w_{\*a_{i}}$ such that $f_{\*X}(\*a_{i}) = w_{\*a_{i}}\min_{\*x \in \*X}d_{\Phi}(\*a_{i},\*x)$. While sampling point (say $\*a_{i}$) the algorithm \okm sets $\*c_{i} = \*a_{i}$ and $\omega_i = w_{\*a_{i}}/p_{i}$ with probability $p_{i}$.

\subsection{Online Coresets for \texorpdfstring{$k$}{k}-means Clustering}
When the divergence $d_{\Phi}()$ is squared Euclidean, the problem is $k$-means clustering. Here, $\forall i \in [n]$ we have $\*M_{i} = \*I_{d}$ and $\mu_{i} = 1$. So the algorithm \okm does not need to maintain $\*M_{i}$ and $\mu_{i}$. In the following corollary we state the guarantee of \okm for k-means clustering.
\begin{corollary}{\label{cor:lwk}}
 Let $\*A \in \~R^{n \times d}$ such that the points are coming in streaming manner and fed to \okm, it returns a coreset $\*C$ which ensures the guarantee in equation \eqref{eq:lwkm} for all $\*X \in \~R^{k \times d}$ with probability at least $0.99$. Such a coreset has $O\Big(\frac{dk\log(1/\epsilon)}{\epsilon^{2}}\big(\log n + \log \big(f_{\varphi}(\*A)\big) - \log \big(f_{\varphi_{2}}(\*a_{2})\big)\big)\Big)$ expected samples. The update time of \okm is $O(d)$ time and uses $O(d)$ as working space.
\end{corollary}
\begin{proof}{\label{proof:lwkbd}}
 The proof follows by combining Lemma \ref{lemma:lwkScore} and Lemma \ref{lemma:lwkGuarantee}. As k-means clustering has $\*M_{i} = \*I_{d}$ and $\mu_{i} = 1$ for all $i \leq n$, hence
 \begin{equation*}
  l_{i} = \frac{f_{\varphi_{i}}(\*a_{i})}{\sum_{j \leq i}f_{\varphi_{j}}(\*a_{j})} + \frac{8}{i-1}
 \end{equation*}
 It can be verified by a similar analysis as in the proof \ref{proof:lwkScore} of Lemma \ref{lemma:lwkScore}. The proof of the second part of Lemma \ref{lemma:lwkScore} and Lemma \ref{lemma:lwkGuarantee} will follow as it is. Further note that k-means is a hard clustering hence the $\epsilon$-net size is $O(2\epsilon^{-dk})$. Hence the expected size of $\*C$ returned by \okm is $O\Big(\frac{dk\log (1/\epsilon)}{\epsilon^{2}}\Big(\log n + \log \big(f_{\varphi}(\*A)\big) - \log \big(f_{\varphi_{2}}(\*a_{2})\big)\Big)\Big)$. At each point the update time is $O(d)$ and uses a working space of $O(d)$.
\end{proof}

%% file: sections/nonparametric.tex
\section{Non-Parametric Coresets for Clustering}
The above coreset size is a function of both $k$, the number of clusters, as well as $d$, the dimension of data points. As discussed before, such dependence restricts the use of such a coreset to settings where a realistic upper bound on $k$ is known, and is supplied as input to the coreset construction. Here we explore the possibility of non-parametric coreset. A coreset is called {\em non-parametric} if the coreset size is independent of $k$ and it can ensure the desired guarantee for any $\*X$ with at most $n$ centres. 

First we state a simple yet important impossibility result. It is not possible to get a non-parametric coreset which ensures a relative error approximation for any clustering problem. Formally we state it in the following theorem,
\begin{theorem}{\label{thm:impossible}}
  There exists a set $\*A$, with $n$ points in $\~R^{d}$ such that there is no $\*C \subset \*A$ with $|\*C| = o(n)$, and $\*C$ which is independent of $k$ (i.e., \#centres), such that for $\epsilon > 0$ and for all $i \in [n]$ it ensures the following with constant probability.
 \begin{equation*}
  |f_{\tilde{\*X}_{i}}(\*A) - f_{\hat{\*X}_{i}}(\*A)| \leq \epsilon f_{\tilde{\*X}_{i}}(\*A)
 \end{equation*}
 where $\tilde{\*X}_{i}$ and $\hat{\*X}_{i}$ are the optimal $i$ centres in $\*A$ and $\*C$.
\end{theorem}

\begin{proof}{\label{proof:impossible}}
 Consider that given $\*A$ such that it has $k_{1}$ natural clusters, i.e., all points highly concentrated in a small radius around its corresponding centres and distance between the centres are significantly high. Now for a non-parametric coreset $\*C$ we expect that, $\forall \*X$ with at most $n$ centres it ensures,
 \begin{equation*}
  |f_{\*X}(\*A) - f_{\*X}(\*C)| \leq \epsilon f_{\*X}(\*A)
 \end{equation*}
 Although note that for our case if $|\*C| < k_{1}$ then for all $i \geq k_{1}$ the optimal centres $\hat{\*X}_{i}$ of the coreset will be $\*C$ itself. Further for coreset size less than $k_{1}$ it is not possible for $f_{\hat{\*X}_{i}}(\*A)$ to ensure a relative approximation to corresponding optimal cost $f_{\tilde{\*X}_{i}}(\*A)$.
\end{proof}

Even for small additive error guarantee (similar to the guarantee in Theorem \ref{thm:lwkm}) it is not clear how to get a non-parametric coreset, due to the union bound over the $\epsilon$-net of the query space. Naively to capture the non-parametric nature the query space has to be a function of $n$, which due to the union bound would reflect in the coreset size.

Once we establish the challenge of a non-parametric coreset, we next present a technique that effectively serves as an existential proof to show that additive error non-parametric coresets exist. This algorithm uses an {\bf oracle} that will return upper and lower barrier sensitivity values when queried. We first show that if such an oracle exists, then we can guarantee the existence of an additive error non-parametric coreset whose size is independent of $d$ and $k$, and depends only on $\log n$ and the structure of the data. 
The guarantees given by the coreset will hold for all set of centres $k \in [n]$. 
As of now, without any assumption, implementing the oracle efficiently remains an open question. Under certain assumption we give an algorithm that returns a non-parametric coreset. We run this method and present it in the experiment section.

To show the existence of non-parametric coreset we combine the sensitivity framework similar to lightweight coresets along with the barrier functions technique from \cite{cohen2016online, batson2012twice} in order to decide the sampling probability of each point. 

We are able to show that the coresets are non-parametric in nature due to the following points. 
\begin{itemize}
  \item The algorithm uses upper bounds to the sensitivity scores. Over expectation these upper bounds are independent of $k$, i.e., the number of centers.
  \item Sampling based on the upper bound returns a strong coreset, i.e., the coreset ensures the guarantee as \eqref{eq:lwkm}, for all $\*X \in \~R^{k \times d}$. 
  \item As the expected sampling complexity is independent of both $k$ and $d$ we further take an union bound over $k \in [n]$. This union bound ensures that the coreset is non-parametric in nature and the guarantee \eqref{eq:lwkm} holds for all $\*X$ with at most $n$ centres in $\~R^{d}$.
\end{itemize}

Unlike \cite{langberg2010universal}, in order to show a strong coreset guarantee we do not need to utilize VC-dimension based arguments. 

\subsection{Coresets for Clustering with Bregman Divergence}
Here we give a theoretical analysis of our existential result using \oneshot. It uses an oracle to create a non-parametric coreset for clustering based on some Bregman divergence in table \ref{tab:bd}. 
The coreset is constructed via importance sampling for which we follow sensitivity based framework along with barrier functions, similar to~\cite{batson2012twice,cohen2016online}.

Let $\*A_{i-1}, \varphi_{i}$ be the same as defined in the previous section. 
Let $\*C_{i-1}$ be a coreset that the algorithm has maintained so far and let $\@X$ is a set with infinitely many elements with every element is some $\*X \in \~R^{k \times d}$, for all $k \leq n$. 
Similar to \cite{cohen2016online, batson2012twice} we also define sensitivity scores using an upper barrier function $(1+\epsilon)f_{\*X}(\*A_{i-1})$ and a lower barrier function $(1-\epsilon)f_{\*X}(\*A_{i-1})$. We informally call these sensitivity scores as upper barrier and lower barrier sensitivity scores. At step $i$, the upper barrier sensitivity $l_i^u$ and the lower barrier sensitivity $l_i^l$ are defined as follows:
\begin{eqnarray}
 l_i^u = \sup_{\*X \in \mathcal{X}} \frac{f_{\*X}(\*a_{i})}{(1+\epsilon)f_{\*X}(\*A_{i-1})-f_{\*X}(\*C_{i-1})+\epsilon f_{\varphi_{i}}(\*A_{i})} \label{eq:uSensitivity} \\
 l_i^l =\sup_{\*X \in \mathcal{X}} \frac{f_{\*X}(\*a_{i})}{f_{\*X}(\*C_{i-1})-(1-\epsilon)f_{\*X}(\*A_{i-1})+\epsilon f_{\varphi_{i}}(\*A_{i})} \label{eq:lSensitivity}
\end{eqnarray}

In Algorithm \ref{alg:lwdkm}, for each $\*a_i$, the sampling probability $p_{i}$ depends on the scores $l_{i}^{u}$ and $l_{i}^{l}$. We consider that the algorithm gets the upper bound of these scores by some oracle. 
The algorithm samples each point with respect to the upper bounds of upper barrier sensitivity score \eqref{eq:uSensitivity} and lower barrier sensitivity scores \eqref{eq:lSensitivity}. 
Note that in the above sensitivity scores the query $\*X$ acts as centers for $\*A_{i-1}$ and $\*C_{i-1}$.

\begin{algorithm}[htpb]
  \caption{\oneshot}{\label{alg:lwdkm}}
  \begin{algorithmic}
  \REQUIRE Input points $\*a_{i}, i = 1,\ldots n; t > 1; \epsilon \in (0,1)$
  \ENSURE $(\text{Coreset } \coreset, \text{Weights } \Omega)$
  \STATE $c^{u} = 2/\epsilon + 1; c^{l} = 2/\epsilon - 1; \varphi_{0} = \emptyset; S = 0; \*C_{0} = \Omega_{0} = \emptyset$
  \STATE $\lambda = \|\*a_{1}\|_{\min}; \quad \nu = \|\*a_{1}\|_{\max}$
  \WHILE {$i \leq n$}
    \STATE $\lambda = \min\{\lambda,\|\*a_{i}\|_{\min}\}; \nu = \max\{\nu,\|\*a_{i}\|_{\max}\}$
    \STATE Update $\*M_{i};\ \mu_{i} = \lambda/\nu$ 
    \STATE $\varphi_{i} = ((i-1)\varphi_{i-1} + \*a_{i})/i ; S = S + f_{\varphi_{i}}^{\*M_{i}}(\*a_{i})$
    \IF{$i = 1$}
      \STATE $p_{i} = 1$
    \ELSE
      \STATE $l_{i}^{u} = \sup_{\@X} \frac{f_{\*X}(\*a_{i})}{(1+\epsilon)f_{\*X}(\*A_{i-1}) - f_{\*X}(\*C_{i-1}) + \epsilon f_{\varphi_{i}}(\*A_{i})}$
      \STATE $l_{i}^{l} = \sup_{\@X} \frac{f_{\*X}(\*a_{i})}{f_{\*X}(\*C_{i-1}) - (1-\epsilon)f_{\*X}(\*A_{i-1}) + \epsilon f_{\varphi_{i}}(\*A_{i})}$
      \STATE $p_{i} = \min\{1,(c^{u}l_{i}^{u}+c^{l}l_{i}^{l})\}$
    \ENDIF
    \STATE Set $\*c_{i}$ and $\omega_{i}$ as \\
    $\begin{cases}\*a_{i} \mbox{ and } 1/p_{i} \qquad \qquad \mbox{w. p. } p_{i} \\  
    \emptyset \mbox{ and } 0 \qquad \qquad \qquad \mbox{else} \end{cases}$
    \STATE $(\*C_{i}, \Omega_{i}) = (\*C_{i-1}, \Omega_{i-1}) \cup (\*c_{i},\omega_{i})$
  \ENDWHILE
  \STATE Return $(\*C, \Omega)$
  \end{algorithmic}
\end{algorithm}

The algorithm maintains a coreset $\*C_{i}$ which ensures a deterministic guarantee as in equation \eqref{eq:lwkm}, $\forall i \in [n]$. We state our guarantee in the following theorem.

\begin{theorem}{\label{thm:lwdkm}}
Let $\*A \in \~R^{n \times d}$ for every Bregman divergence $d_{\Phi}$ as in table \ref{tab:bd} there exists a coresets $\*C$ for clustering based on $d_{\Phi}$ such that the following statement is ensured for all $\*X$ with at most $n$ centres in $\~R^{d}$,
 \begin{equation}{\label{eq:oneshot}}
 \Big|f_{\*X}(\*C) - f_{\*X}(\*A)\Big| \leq \epsilon(f_{\*X}(\*A) + f_{\varphi}(\*A))
\end{equation}
 Such coreset has $O\Big(\frac{\log n}{\mu\epsilon^{2}}\Big(\log n + \log \big(f_{\varphi}^{\*M}(\*A)\big) - \log \big(f_{\varphi_{2}}^{\*M_{2}}(\*a_{2})\big)\Big)\Big)$ expected samples.
\end{theorem}


We prove the above theorem with the following supporting lemmas. 
We first show that for each point $\*a_{i}$, if $l_{i}^{u}$ and $l_{i}^{l}$ upper bound the sensitivity scores \eqref{eq:uSensitivity} and \eqref{eq:lSensitivity} respectively, then the coreset returned by \oneshot~ensures the guarantee as \eqref{eq:lwkm}.

\begin{lemma}{\label{lemma:lwdkmGuarantee}}
 Suppose the scores $l_{i}^{u}$ and $l_{i}^{l}$ received by oracle in \oneshot~upper bound both scores \eqref{eq:uSensitivity} and \eqref{eq:lSensitivity} respectively $\forall i \in [n]$ with $\*C_{i-1}$. \oneshot~computes the sampling probability for the $i^{th}$ point as $p_{i} = \min\{\tilde{l}_{i},1\}$ where $\tilde{l}_{i} = c^{u}l_{i}^{u} + c^{l}l_{i}^{l}$.
 The statement \eqref{eq:lwkm} is then true for every $i \in [n]$, for all $\*X$ with at most $n$ centres in $\~R^{d}$.
\end{lemma}
\begin{proof}{\label{proof:lwdkmGuarantee}}
 We show this by induction. The proof applies $\forall \*X$ with at most $n$ centres in $\~R^{d}$. Now for $i=1$ this is trivially true, as we have $p_{1} = 1$. So we we have $\*c_{1} = \*a_{1}$ and hence we get,
 \begin{equation}
  (1 - \epsilon)f_{\*X}(\*a_{1}) \leq f_{\*X}(\*c_{1}) \leq (1 + \epsilon)f_{\*X}(\*a_{1})
 \end{equation}
 Consider that at $i-1$ the $\coreset_{i-1}$ ensures the following,
 \begin{equation}
  (1 - \epsilon)f_{\*X}(\*A_{i-1}) - \epsilon f_{\varphi_{i-1}}(\*A_{i-1}) \leq f_{\*X}(\*C_{i-1}) \leq (1 + \epsilon)f_{\*X}(\*A_{i-1}) + \epsilon f_{\varphi_{i-1}}(\*A_{i-1})
 \end{equation}
 Now we show inductively for $\*C_i$. Here the sampling probability $p_{i} = \min\{1,c^{u}l_{i}^{u}+c^{l}l_{i}^{l}\}$ if $p_{i} = 1$ then the following is true,
 \begin{eqnarray*}
  (1 - \epsilon)f_{\*X}(\*A_{i-1}) - \epsilon f_{\varphi_{i-1}}(\*A_{i-1}) &\leq f_{\*X}(\*C_{i-1}) &\leq (1 + \epsilon)f_{\*X}(\*A_{i-1}) + \epsilon f_{\varphi_{i-1}}(\*A_{i-1}) \\
  (1 - \epsilon)f_{\*X}(\*A_{i-1}) - \epsilon f_{\varphi_{i-1}}(\*A_{i-1}) + f_{\*X}(\*a_{i}) &\leq f_{\*X}(\*C_{i-1}) + f_{\*X}(\*a_{i}) &\leq (1 + \epsilon)f_{\*X}(\*A_{i-1}) + \epsilon f_{\varphi_{i-1}}(\*A_{i-1}) + f_{\*X}(\*a_{i}) \\
  (1 - \epsilon)f_{\*X}(\*A_{i}) - \epsilon f_{\varphi_{i}}(\*A_{i}) &\leq f_{\*X}(\*C_{i}) &\leq (1 + \epsilon)f_{\*X}(\*A_{i}) + \epsilon f_{\varphi_{i}}(\*A_{i})
 \end{eqnarray*}
 Note that $f_{\varphi_{i}}(\*A_{i}) \geq f_{\varphi_{i-1}}(\*A_{i-1})$. Now if $p_{i} < 1$ then for the upper barrier function we use the definition of $l_{i}^{u}$,
 \begin{eqnarray*}
  p_{i} &\geq& l_{i}^{u} \\
  p_{i} &\geq& \frac{f_{\*X}(\*a_{i})}{(1+\epsilon)f_{\*X}(\*A_{i-1}) - f_{\*X}(\*C_{i-1}) + \epsilon f_{\varphi_{i}}(\*A_{i})} \\
  (1+\epsilon)f_{\*X}(\*A_{i-1}) - f_{\*X}(\*C_{i-1}) + \epsilon f_{\varphi_{i}}(\*A_{i}) &\geq& \frac{f_{\*X}(\*a_{i})}{p_{i}} \\
  (1+\epsilon)f_{\*X}(\*A_{i-1}) + \epsilon f_{\varphi_{i}}(\*A_{i}) &\geq& f_{\*X}(\*C_{i-1}) + \frac{f_{\*X}(\*a_{i})}{p_{i}} \\
  (1+\epsilon)f_{\*X}(\*A_{i-1}) + \epsilon f_{\varphi_{i}}(\*A_{i}) &\geq& f_{\*X}(\*C_{i}) \\
  (1+\epsilon)f_{\*X}(\*A_{i}) + \epsilon f_{\varphi_{i}}(\*A_{i}) &\geq& f_{\*X}(\*C_{i})
 \end{eqnarray*}
 Note that if $\*a_{i}$ is not sampled the RHS will be smaller. The above analysis shows that the upper barrier claim in the lemma holds even for $i^{th}$ point. Next for the lower barrier we use the definition of $l_{i}^{l}$,
 \begin{eqnarray*}
  1 &>& l_{i}^{l} \\
  1 &\geq& \frac{f_{\*X}(\*a_{i})}{f_{\*X}(\*C_{i-1}) - (1-\epsilon)f_{\*X}(\*A_{i-1}) + \epsilon f_{\varphi_{i}}(\*A_{i})} \\
  f_{\*X}(\*C_{i-1}) - (1-\epsilon)f_{\*X}(\*A_{i-1}) + \epsilon f_{\varphi_{i}}(\*A_{i}) &\geq& f_{\*X}(\*a_{i}) \\
  f_{\*X}(\*C_{i-1}) &\geq& (1-\epsilon)f_{\*X}(\*A_{i-1}) - \epsilon f_{\varphi_{i}}(\*A_{i}) + f_{\*X}(\*a_{i}) \\
  f_{\*X}(\*C_{i-1}) &\geq& (1-\epsilon)f_{\*X}(\*A_{i}) - \epsilon f_{\varphi_{i}}(\*A_{i}) \\
  f_{\*X}(\*C_{i}) &\geq& (1-\epsilon)f_{\*X}(\*A_{i}) - \epsilon f_{\varphi_{i}}(\*A_{i})
 \end{eqnarray*}
 Note that the above is true when the point $\*a_{i}$ is not sampled in the coreset. If $\*a_{i}$ is sampled then the LHS will be bigger. Hence the above analysis shows that the lower barrier claim in the lemma holds even for $i^{th}$ point.
\end{proof}
The above lemma ensures a deterministic guarantee. It is important to note that due to the barrier function based sampling, the guarantee stands $\forall \*X$, and we do not require the knowledge of pseudo dimension of the query space. This is similar to the deterministic spectral sparsification claim of~\cite{batson2012twice}. Now we discuss a supporting lemma which we use to bound the expected sample size.
 \begin{lemma}{\label{lemma:lwdksupport}}
 Given scalars $q, r, s, u, v$ and $w$, where $q, r, s$ and $w$ are positive, we define a random variable $t$ as,
 \begin{equation*}
  t = 
  \begin{cases} 
  q - u\cdot r \qquad \qquad \mbox{with probability } p, \\
  q - v\cdot r \qquad \qquad \mbox{with probability } (1-p). \end{cases}
 \end{equation*}
 Then if $\frac{r}{q+w} = 1$ we get,
 \begin{equation*}
  \~E\bigg[\frac{s}{t + w} - \frac{s}{q + w}\bigg] = \frac{pu+(1-p)v-uv}{(1-u)(1-v)}\bigg(\frac{s}{q + w}\bigg)
 \end{equation*}
\end{lemma}
\begin{proof}{\label{proof:lwdksupport}}
 The proof is fairly straight forward. Using simple algebra (similar to  \cite{sherman1950adjustment}) we have,
 \begin{eqnarray*}
  \frac{1}{q + w - ur} &=& \frac{1}{q + w} + \frac{ur(q + w)^{-2}}{1-ur(q + w)^{-1}} \\
  &=& \frac{1}{q + w} + \frac{u}{1-u}(q + w)^{-1}
 \end{eqnarray*}
 \begin{eqnarray*}
  \frac{1}{q + w - vr} &=& \frac{1}{q + w} + \frac{vr(q + w)^{-2}}{1-vr(q + w)^{-1}} \\
  &=& \frac{1}{q + w} + \frac{v}{1-v}(q + w)^{-1}
 \end{eqnarray*}
 So we get,
 \begin{equation*}
  \~E\bigg[\frac{s}{t + w} - \frac{s}{q + w}\bigg] = \frac{pu+(1-p)v-uv}{(1-u)(1-v)}\bigg(\frac{s}{q + w}\bigg)
 \end{equation*}
\end{proof}
Let $\*C_{i-1}$ be the coreset at point $i-1$. Let $\pi_{i-1}$ be the sampling/no-sampling choices that \oneshot made while creating $\*C_{i-1}$. Let upper and lower barrier sensitivity scores be $l_{i}^{u}$ and $l_{i}^{l}$ respectively, which depend on the coreset maintained so far. 
\begin{eqnarray}
 \~E_{\pi_{i-1}}[l_{i}^{u}] = \~E_{\pi_{i-1}}\bigg[\sup_{\*X} \frac{f_{\*X}(\*a_{i})}{(1+\epsilon)f_{\*X}(\*A_{i-1}) - f_{\*X}(\*C_{i-1}) + \epsilon f_{\varphi_{i}}(\*A_{i})}\bigg] \\
 \~E_{\pi_{i-1}}[l_{i}^{l}] = \~E_{\pi_{i-1}}\bigg[\sup_{\*X} \frac{f_{\*X}(\*a_{i})}{f_{\*X}(\*C_{i-1}) - (1-\epsilon)f_{\*X}(\*A_{i-1}) + \epsilon f_{\varphi_{i}}(\*A_{i})}\bigg]
\end{eqnarray}

\begin{lemma}{\label{lemma:lwdkScore}}
 For all $\*X \in \~R^{k \times d}$ and for all $i \in [n]$, we have, 
 \begin{eqnarray*}
  \~E_{\pi_{i-1}}[l_{i}^{u}] &\leq& \frac{2f_{\varphi_{i}}^{\*M_{i}}(\*a_{i})}{\mu_{i}\epsilon \sum_{j \leq i}f_{\varphi_{j}}^{\*M_{j}}(\*a_{j})} + \frac{12}{\mu_{i}\epsilon(i-1)} \\
  \~E_{\pi_{i-1}}[l_{i}^{l}] &\leq& \frac{2f_{\varphi_{i}}^{\*M_{i}}(\*a_{i})}{\mu_{i}\epsilon \sum_{j \leq i}f_{\varphi_{j}}^{\*M_{j}}(\*a_{j})} + \frac{12}{\mu_{i}\epsilon(i-1)}
 \end{eqnarray*}
 where $\*M_{i}$ and $\mu_{i}$ is defined for $\*A_{i}$ and for specific Bregman divergence such that $\forall \*x \in \~R^{d}, \*a \in \*A_{i}$ we have $\mu_{i}f_{\*x}^{\*M_{i}}(\*a) \leq f_{\*x}(\*a) \leq f_{\*x}^{\*M_{i}}(\*a)$.
\end{lemma}
\begin{proof}{\label{proof:lwdkScore}}
  Let $\*A_{i}$ represents the first $i$ points. Further in this proof $\*X \in \~R^{k \times d}$. For a fixed $\*X$ we define two scalars $(\zeta_{\*X})_{i,j}^{u}$ and $(\zeta_{\*X})_{i,j}^{l}$ as follows,
 \begin{eqnarray*}
  (\zeta_{\*X})_{i,j}^{u} &=& \frac{\epsilon}{2}f_{\*X}(\*A_{i}) + (1+\frac{\epsilon}{2})f_{\*X}(\*A_{j}) \\
  (\zeta_{\*X})_{i,j}^{l} &=& -\frac{\epsilon}{2}f_{\*X}(\*A_{i}) + (1-\frac{\epsilon}{2})f_{\*X}(\*A_{j})
 \end{eqnarray*}
 So we have $(\zeta_{\*X})_{i,i}^{u} = (1+\epsilon)f_{\*X}(\*A_{i})$ and $(\zeta_{\*X})_{i,i}^{l} = (1-\epsilon)f_{\*X}(\*A_{i})$. It is clear that for $j \leq i-1$ we have $(\zeta_{\*X})_{i-1,j}^{u} \geq (\zeta_{\*X})_{j,j}^{u}$ and $(\zeta_{\*X})_{i-1,j}^{l} \leq (\zeta_{\*X})_{j,j}^{l}$. Further two more scalars $(\gamma_{\*X})_{i,j}^{u}$ and $(\gamma_{\*X})_{i,j}^{l}$ are defined as follows,
 \begin{eqnarray*}
  (\gamma_{\*X})_{i,j}^{u} &=& (\zeta_{\*X})_{i,j}^{u} - f_{\*X}(\*C_{j}) \\
  (\gamma_{\*X})_{i,j}^{l} &=& f_{\*X}(\*C_{j}) - (\zeta_{\*X})_{i,j}^{l}
 \end{eqnarray*}
 Note that $(\gamma_{\*X})_{i,i}^{u} = (1+\epsilon)f_{\*X}(\*A_{i}) - f_{\*X}(\*C_{i})$ and $(\gamma_{\*X})_{i,i}^{l} = f_{\*X}(\*C_{i}) - (1-\epsilon)f_{\*X}(\*A_{i})$. For $j \leq i-1$ we get $(\gamma_{\*X})_{i-1,j}^{u} \geq (\gamma_{\*X})_{j,j}^{u}$ and $(\gamma_{\*X})_{i-1,j}^{l} \geq (\gamma_{\*X})_{j,j}^{l}$. Let, $(d_{\*X})_{j+1} = \frac{f_{\*X}(\*a_{j+1})}{p_{j+1}}$. If $p_{j+1} < 1$, then we have $p_{j+1} \ge c_u l_{j+1}^u$, and hence we have the following for upper barrier,
 \begin{eqnarray*}
  p_{j+1} &\geq& \frac{c^{u}f_{\*X}(\*a_{j+1})}{(\gamma_{\*X})_{j,j}^{u} + f_{\varphi_{j+1}}(\*A_{j+1})} \\
  &\geq& \frac{c^{u}f_{\*X}(\*a_{j+1})}{(\gamma_{\*X})_{i-1,j}^{u} + f_{\varphi_{j+1}}(\*A_{j+1})} \\
  &\geq& \frac{c^{u}f_{\*X}(\*a_{j+1})}{(\gamma_{\*X})_{i-1,j}^{u} + f_{\varphi_{i}}(\*A_{i})} \\
  \frac{(d_{\*X})_{j+1}}{(\gamma_{\*X})_{i-1,j}^{u} + f_{\varphi_{i}}(\*A_{i})} &\leq& \frac{1}{c^{u}} 
 \end{eqnarray*}
Let $\frac{(d_{\*X})_{j+1}}{(\gamma_{\*X})_{i-1,j}^{u} + f_{\varphi_{i}}(\*A_{i})} = (h_{\*X})_{j+1}^{u}$, which is bounded by $\frac{1}{c^{u}}$. Similarly for the lower barrier we have,
 \begin{eqnarray*}
  p_{j+1} &\geq& \frac{c^{l}f_{\*X}(\*a_{j+1})}{(\gamma_{\*X})_{j,j}^{l} + f_{\varphi_{j+1}}(\*A_{j+1})} \\
  &\geq& \frac{c^{l}f_{\*X}(\*a_{j+1})}{(\gamma_{\*X})_{i-1,j}^{l} + f_{\varphi_{j+1}}(\*A_{j+1})} \\
  &\geq& \frac{c^{l}f_{\*X}(\*a_{j+1})}{(\gamma_{\*X})_{i-1,j}^{l} + f_{\varphi_{i}}(\*A_{i})} \\
  \frac{(d_{\*X})_{j+1}}{(\gamma_{\*X})_{i-1,j}^{l} + f_{\varphi_{i}}(\*A_{i})} &\leq& \frac{1}{c^{l}} 
 \end{eqnarray*}
 Let $\frac{(d_{\*X})_{j+1}}{(\gamma_{\*X})_{i-1,j}^{l} + f_{\varphi_{i}}(\*A_{i})} = (h_{\*X})_{j+1}^{l}$, which is bounded by $\frac{1}{c^{l}}$. Next we apply the Lemma \ref{lemma:lwdksupport} to get an upper bound on the sensitivity scores i.e.,
 \begin{eqnarray*}
  \sup_{\*X} \frac{f_{\*X}(\*a_{i})}{(1+\epsilon)f_{\*X}(\*A_{i-1}) - f_{\*X}(\*C_{i-1}) + f_{\varphi_{i}}(\*A_{i})}  \\
  \sup_{\*X} \frac{f_{\*X}(\*a_{i})}{f_{\*X}(\*C_{i-1}) - (1-\epsilon)f_{\*X}(\*A_{i-1}) + f_{\varphi_{i}}(\*A_{i})} 
 \end{eqnarray*}

 We apply Lemma \ref{lemma:lwdksupport} and set $q = (\gamma_{\*X})_{i-1,j}^{u}, r = (d_{\*X})_{j+1}/(h_{\*X})_{j+1}^{u}, s = f_{\*X}(\*a_{i})$ and $w = \epsilon f_{\varphi_{i}}(\*A_{i})$. Further let $u = (h_{\*X})_{j+1}^{u}(1 - p_{j+1}(1 + \epsilon/2)), v = -(h_{\*X})_{j+1}^{u}p_{j+1}(1 + \epsilon/2)$ and $p = p_{j+1}$.
 
 Note that with the above substitution we have $\frac{r}{q + w} = 1$ and $t = (\gamma_{\*X})_{i-1,j+1}^{u}$. Further with $c^{u} \geq 2/\epsilon + 1$ we also have the RHS of the lemma \ref{lemma:lwdksupport}, $\frac{pu + (1-p)v - uv}{(1-u)(1-v)} \leq 0$, So we have,
 \begin{equation*}
 \~E_{\pi_{j+1}}\bigg[\frac{f_{\*X}(\*a_{i})}{(\gamma_{\*X})_{i-1,j+1}^{u} + \epsilon f_{\varphi_{i}}(\*A_{i})} \bigg] \leq \~E_{\pi_{j}}\bigg[\frac{f_{\*X}(\*a_{i})}{(\gamma_{\*X})_{i-1,j}^{u} + \epsilon f_{\varphi_{i}}(\*A_{i})} \bigg]
\end{equation*}

Let $(1+\epsilon)f_{\*X}(\*A_{i-1}) - f_{\*X}(\*C_{i-1}) = f_{\*X}(\*A_{i-1}^{u})$ where $f_{\*X}(\*A_{i-1}^{u}) = \sum_{j \leq i-1} f_{\*X}(\*a_{j}^{u})$. Here each term $f_{\*X}(\*a_{j}^{u}) = (1+\epsilon-p_{j}^{-1})f_{\*X}(\*a_{j})$ if $\*a_{j}$ is present in $\*C_{i-1}$ else $f_{\*X}(\*a_{j}^{u}) = (1+\epsilon)f_{\*X}(\*a_{j})$. 
Now the upper sensitivity score with respect to $\*X \in \~R^{k \times d}$ can be bounded as follows. 

 \begin{eqnarray*}
  \~E_{\pi_{i-1}}\bigg[\frac{f_{\*X}(\*a_{i})}{f_{\*X}(\*A_{i-1}^{u}) + \epsilon f_{\varphi_{i}}(\*A_{i})}\bigg] &\stackrel{(i)}{\leq}& \~E_{\pi_{i-1}}\bigg[\frac{f_{\*X}^{\*M_{i}}(\*a_{i})}{f_{\*X}(\*A_{i-1}^{u}) + \epsilon f_{\varphi_{i}}(\*A_{i})}\bigg] \\
  &\stackrel{(ii)}{\leq}& \~E_{\pi_{i-1}}\bigg[\frac{\Big[2f_{\varphi_{i}}^{\*M_{i}}(\*a_{i}) + \frac{4}{i-1}\sum_{\*a_{j} \in \*A_{i-1}}[f_{\varphi_{i}}^{\*M_{i}}(\*a_{j}) + f_{\*X}^{\*M_{i}}(\*a_{j})]\Big]}{(1+\epsilon)f_{\*X}(\*A_{i-1}) - f_{\*X}(\*C_{i-1}) + \epsilon f_{\varphi_{i}}(\*A_{i})} \Big| \pi_{i-2}\bigg] \\
  &=& \~E_{\pi_{i-1}}\bigg[\frac{\Big[2f_{\varphi_{i}}^{\*M_{i}}(\*a_{i}) + \frac{4}{i-1}\sum_{\*a_{j} \in \*A_{i-1}}f_{\varphi_{i}}^{\*M_{i}}(\*a_{j})\Big]}{(1+\epsilon)f_{\*X}(\*A_{i-1}) - f_{\*X}(\*C_{i-1}) + \epsilon f_{\varphi_{i}}(\*A_{i})} \Big| \pi_{i-2}\bigg] \\
  &+& \~E_{\pi_{i-1}}\bigg[\frac{\frac{4}{i-1}f_{\*X}^{\*M_{i}}(\*A_{i-1})}{(1 + \epsilon)f_{\*X}(\*A_{i-1}) - f_{\*X}(\*C_{i-1}) + \epsilon f_{\varphi_{i}}(\*A_{i})} \Big| \pi_{i-2}\bigg] \\
  &\stackrel{(iii)}{=}& \~E_{\pi_{i-1}}\bigg[\frac{2f_{\varphi_{i}}^{\*M_{i}}(\*a_{i}) + \frac{4}{i-1}f_{\varphi_{i}}^{\*M_{i}}(\*A_{i-1})} {(\gamma_{\*X})_{i-1,i-1}^{u} + \epsilon f_{\varphi_{i}}(\*A_{i})} \Big| \pi_{i-2}\bigg] + \~E_{\pi_{i-1}}\bigg[\frac{\frac{4}{i-1}f_{\*X}^{\*M_{i}}(\*A_{i-1})}{(\gamma_{\*X})_{i-1,i-1}^{u} + \epsilon f_{\varphi_{i}}(\*A_{i})} \Big| \pi_{i-2}\bigg] \\
  &\stackrel{(iv)}{\leq}& \~E_{\pi_{i-2}}\bigg[\frac{2f_{\varphi_{i}}^{\*M_{i}}(\*a_{i}) + \frac{4}{i-1}f_{\varphi_{i}}^{\*M_{i}}(\*A_{i-1})} {(\gamma_{\*X})_{i-1,i-2}^{u} + \epsilon f_{\varphi_{i}}(\*A_{i})} \Big| \pi_{i-3}\bigg] + \~E_{\pi_{i-2}}\bigg[\frac{\frac{4}{i-1}f_{\*X}^{\*M_{i}}(\*A_{i-1})}{(\gamma_{\*X})_{i-1,i-2}^{u} + \epsilon f_{\varphi_{i}}(\*A_{i})} \Big| \pi_{i-3}\bigg] \\
  &\stackrel{(v)}{\leq}& \~E_{\pi_{0}}\bigg[\frac{2f_{\varphi_{i}}^{\*M_{i}}(\*a_{i}) + \frac{4}{i-1} f_{\varphi_{i}}^{\*M_{i}}(\*A_{i-1})} {(\gamma_{\*X})_{i-1,0}^{u} + \epsilon f_{\varphi_{i}}(\*A_{i})}\bigg] + \~E_{\pi_{0}}\bigg[\frac{\frac{4}{i-1}f_{\*X}^{\*M_{i}}(\*A_{i-1})}{(\gamma_{\*X})_{i-1,0}^{u}+\epsilon f_{\varphi_{i}}(\*A_{i})}\bigg] \\
  &=& \frac{2f_{\varphi_{i}}^{\*M_{i}}(\*a_{i}) + \frac{4}{i-1}f_{\varphi_{i}}^{\*M_{i}}(\*A_{i-1})}{\epsilon/2f_{\*X}(\*A_{i-1}) + \epsilon f_{\varphi_{i}}(\*A_{i})} + \frac{\frac{4}{(i-1)}f_{\*X}^{\*M_{i}}(\*A_{i-1})}{\epsilon/2f_{\*X}(\*A_{i-1}) + \epsilon f_{\varphi_{i}}(\*A_{i})} \\
  &\stackrel{(vi)}{\leq}& \frac{2f_{\varphi_{i}}^{\*M_{i}}(\*a_{i}) + \frac{4}{i-1}f_{\varphi_{i}}^{\*M_{i}}(\*A_{i-1})} {\mu_{i}\epsilon(0.5f_{\*X}^{\*M_{i}}(\*A_{i-1}) + f_{\varphi_{i}}^{\*M_{i}}(\*A_{i}))} + \frac{\frac{4}{(i-1)}f_{\*X}^{\*M_{i}}(\*A_{i-1})}{\mu_{i}\epsilon(0.5f_{\*X}^{\*M_{i}}(\*A_{i-1})+f_{\varphi_{i}}^{\*M_{i}}(\*A_{i}))} \\
  &\leq& \frac{2f_{\varphi_{i}}^{\*M_{i}}(\*a_{i})}{\mu_{i}\epsilon f_{\varphi_{i}}^{\*M_{i}}(\*A_{i})} + \frac{4f_{\varphi_{i}}^{\*M_{i}}(\*A_{i-1})}{\mu_{i}\epsilon(i-1) f_{\varphi_{i}}^{\*M_{i}}(\*A_{i})} + \frac{8f_{\*X}^{\*M_{i}}(\*A_{i-1})}{\mu_{i}\epsilon(i-1)f_{\*X}^{\*M_{i}}(\*A_{i-1})} \\
  &\stackrel{(vii)}{\leq}& \frac{2f_{\varphi_{i}}^{\*M_{i}}(\*a_{i})}{\mu_{i}\epsilon f_{\varphi_{i}}^{\*M_{i}}(\*A_{i})} + \frac{4}{\mu_{i}\epsilon(i-1)} + \frac{8}{\mu_{i}\epsilon(i-1)} \\
  &\leq& \frac{2f_{\varphi_{i}}^{\*M_{i}}(\*a_{i})}{\mu_{i}\epsilon \sum_{j \leq i}f_{\varphi_{j}}^{\*M_{j}}(\*a_{j})} + \frac{12}{\mu_{i}\epsilon(i-1)}
 \end{eqnarray*}

 The inequality $(i)$ is by upper bounding Bregman divergence by squared Mahalanobis distance. The inequality $(ii)$ is due to applying triangle inequality on the numerator. The $(iii)$ equality is by replacing the denominator with the above assumption. The $(iv)$ inequality is by applying the supporting Lemma \ref{lemma:lwdksupport}. By recursively applying Lemma \ref{lemma:lwdksupport} we get the inequality $(v)$ which is independent of the random choices made by \oneshot. The inequality $(vi)$ is by using the lower bound on the denominator. The inequality $(vii)$ an upper bound on the second and the third term. In the final inequality we use the fact that for any $\mu$ similar Bregman divergence from \cite{ackermann2009coresets, lucic2016strong} we have $\*M_{j} \preceq \*M_{i}$ for $j \leq i$. Further by the property of Bregman divergence we know that $f_{\varphi_{i}}(\*A_{i-1}) \geq f_{\varphi_{i-1}}(\*A_{i-1})$. Hence we have $f_{\varphi_{i}}^{\*M_{i}}(\*A_{i}) \geq \sum_{j\leq i}f_{\varphi_{j}}^{\*M_{j}}(\*a_{j})$. Now note that we have this upper bound for all $\*X \in \~R^{k \times d}$, which is independent of $k$. So we have,
 \begin{equation*}
  \~E_{\pi_{i-1}}\bigg[\frac{f_{\*X}(\*a_{i})}{(1+\epsilon)f_{\*X}(\*A_{i-1}) - f_{\*X}(\*C_{i-1}) + f_{\varphi_{i}}(\*A_{i})}\bigg] \leq \frac{2f_{\varphi_{i}}(\*a_{i})}{\epsilon \sum_{j \leq i}f_{\varphi_{j}}(\*a_{j})} + \frac{12}{\epsilon(i-1)}
 \end{equation*}

 Now for the lower barrier first apply Lemma \ref{lemma:lwdksupport} by setting $q = (\gamma_{\*X})_{i-1,j}^{l}, r = (d_{\*X})_{j+1}/(h_{\*X})_{j+1}^{l}, s = f_{\*X}(\*a_{i})$ and $w = \epsilon f_{\varphi_{i}}(\*A_{i})$. Further let $u = -(h_{\*X})_{j+1}^{l}(1-p_{j+1}(1-\epsilon/2)), v = (h_{\*X})_{j+1}^{l}p_{j+1}(1-\epsilon/2))$ and $p = p_{j+1}$ we get,
 \begin{equation*}
  \~E_{\pi_{j+1}}\bigg[\frac{f_{\*X}(\*a_{i})}{(\gamma_{\*X})_{i-1,j+1}^{l} + \epsilon f_{\varphi_{i}}(\*A_{i})} \bigg] \leq \~E_{\pi_{j}}\bigg[\frac{f_{\*X}(\*a_{i})}{(\gamma_{\*X})_{i-1,j}^{l} + \epsilon f_{\varphi_{i}}(\*A_{i})} \bigg]
 \end{equation*}

 Let $(f_{\*X}(\*C_{i-1}) - (1-\epsilon)f_{\*X}(\*A_{i-1}) = f_{\*X}(\*A_{i-1}^{l})$ where $f_{\*X}(\*A_{i-1}^{l}) = \sum_{j \leq i-1} f_{\*X}(\*a_{j}^{l})$. Here each term $f_{\*X}(\*a_{j}^{l}) = (p_{j}^{-1}-1+\epsilon)f_{\*X}(\*a_{j})$ if $\*a_{j}$ is present in $\*C_{i-1}$ else $f_{\*X}(\*a_{j}^{l}) = (-1+\epsilon)f_{\*X}(\*a_{j})$. 

 \begin{eqnarray*}
  \~E_{\pi_{i-1}}\bigg[\frac{f_{\*X}(\*a_{i})}{f_{\*X}(\*A_{i-1}^{l}) + \epsilon f_{\varphi_{i}}(\*A_{i})}\bigg] &\stackrel{(i)}{\leq}& \~E_{\pi_{i-1}}\bigg[\frac{f_{\*X}^{\*M_{i}}(\*a_{i})}{f_{\*X}(\*A_{i-1}^{l}) + \epsilon f_{\varphi_{i}}(\*A_{i})}\bigg]\\ 
  &\stackrel{(ii)}{\leq}& \~E_{\pi_{i-1}}\bigg[\frac{\Big[2f_{\varphi_{i}}^{\*M_{i}}(\*a_{i}) + \frac{4}{i-1}\sum_{\*a_{j} \in \*A_{i-1}}[f_{\varphi_{i}}^{\*M_{i}}(\*a_{j}) + f_{\*X}^{\*M_{i}}(\*a_{j})]\Big]}{f_{\*X}(\*C_{i-1}) - (1-\epsilon)f_{\*X}^{\*M_{i}}(\*A_{i-1}) + \epsilon f_{\varphi_{ii}}(\*A_{i})} \Big| \pi_{i-2}\bigg] \\
  &=& \~E_{\pi_{i-1}}\bigg[\frac{\Big[2f_{\varphi_{i}}^{\*M_{i}}(\*a_{i}) + \frac{4}{i-1}f_{\varphi_{i}}^{\*M_{i}}(\*A_{i-1})\Big]}{f_{\*X}(\*C_{i-1}) - (1-\epsilon)f_{\*X}(\*A_{i-1}) + \epsilon f_{\varphi_{i}}(\*A_{i})} \Big| \pi_{i-2}\bigg] \\
  &+& \~E_{\pi_{i-1}}\bigg[\frac{\frac{4}{\epsilon(i-1)}f_{\*X}^{\*M_{i}}(\*A_{i-1})}{f_{\*X}(\*C_{i-1}) - (1-\epsilon)f_{\*X}(\*A_{i-1}) + \epsilon f_{\varphi_{i}}(\*A_{i})} \Big| \pi_{i-2}\bigg] \\
  &\stackrel{(iii)}{=}& \~E_{\pi_{i-1}}\bigg[\frac{2f_{\varphi_{i}}^{\*M_{i}}(\*a_{i}) + \frac{4}{i-1}f_{\varphi_{i}}^{\*M_{i}}(\*A_{i-1})} {(\gamma_{\*X})_{i-1,i-1}^{l} + \epsilon f_{\varphi_{i}}(\*A_{i})} \Big| \pi_{i-2}\bigg] + \~E_{\pi_{i-1}}\bigg[\frac{\frac{4}{i-1}f_{\*X}^{\*M_{i}}(\*A_{i-1})}{(\gamma_{\*X})_{i-1,i-1}^{l} + \epsilon f_{\varphi_{i}}(\*A_{i})} \Big| \pi_{i-2}\bigg] \\
  &\stackrel{(iv)}{\leq}& \~E_{\pi_{i-2}}\bigg[\frac{2f_{\varphi_{i}}^{\*M_{i}}(\*a_{i}) + \frac{4}{i-1}f_{\varphi_{i}}^{\*M_{i}}(\*A_{i-1})} {(\gamma_{\*X})_{i-1,i-2}^{l} + \epsilon f_{\varphi_{i}}(\*A_{i})} \Big| \pi_{i-3}\bigg] + \~E_{\pi_{i-3}}\bigg[\frac{\frac{4}{i-1}f_{\*X}^{\*M_{i}}(\*A_{i-1})}{(\gamma_{\*X})_{i-1,i-2}^{l} + \epsilon f_{\varphi_{i}}(\*A_{i})} \Big| \pi_{i-3}\bigg] \\
  &\stackrel{(v)}{\leq}& \~E_{\pi_{0}}\bigg[\frac{2f_{\varphi_{i}}^{\*M_{i}}(\*a_{i}) + \frac{4}{i-1}f_{\varphi_{i}}^{\*M_{i}}(\*A_{i-1})} {(\gamma_{\*X})_{i-1,i-2}^{l} + \epsilon f_{\varphi_{i}}(\*A_{i})}\bigg] + \~E_{\pi_{0}}\bigg[\frac{\frac{4}{i-1}f_{\*X}^{\*M_{i}}(\*A_{i-1})} {(\gamma_{\*X})_{i-1,i-2}^{l} + \epsilon f_{\varphi_{i}}(\*A_{i})}\bigg] \\
  &\stackrel{(vi)}{\leq}& \frac{2f_{\varphi_{i}}^{\*M_{i}}(\*a_{i}) + \frac{4}{i-1}f_{\varphi_{i}}^{\*M_{i}}(\*A_{i-1})} {\mu_{i}\epsilon(0.5f_{\*X}^{\*M_{i}}(\*A_{i-1}) + f_{\varphi_{i}}^{\*M_{i}}(\*A_{i}))} + \frac{\frac{4}{(i-1)}f_{\*X}^{\*M_{i}}(\*A_{i-1})}{\mu_{i}\epsilon(0.5f_{\*X}^{\*M_{i}}(\*A_{i-1})+f_{\varphi_{i}}^{\*M_{i}}(\*A_{i}))} \\
  &\leq& \frac{2f_{\varphi_{i}}^{\*M_{i}}(\*a_{i})}{\mu_{i}\epsilon f_{\varphi_{i}}^{\*M_{i}}(\*A_{i})} + \frac{4f_{\varphi_{i}}^{\*M_{i}}(\*A_{i-1})}{\mu_{i}\epsilon(i-1) f_{\varphi_{i}}^{\*M_{i}}(\*A_{i})} + \frac{8f_{\*X}^{\*M_{i}}(\*A_{i-1})}{\mu_{i}\epsilon(i-1)f_{\*X}^{\*M_{i}}(\*A_{i-1})} \\
  &\stackrel{(vii)}{\leq}& \frac{2f_{\varphi_{i}}^{\*M_{i}}(\*a_{i})}{\mu_{i}\epsilon f_{\varphi_{i}}^{\*M_{i}}(\*A_{i})} + \frac{4}{\mu_{i}\epsilon(i-1)} + \frac{8}{\mu_{i}\epsilon(i-1)} \\
  &\leq& \frac{2f_{\varphi_{i}}^{\*M_{i}}(\*a_{i})}{\mu_{i}\epsilon \sum_{j \leq i}f_{\varphi_{j}}^{\*M_{j}}(\*a_{j})} + \frac{12}{\mu_{i}\epsilon(i-1)}
 \end{eqnarray*}
 
 The inequality $(i)$ is by upper bounding Bregman divergence by squared Mahalanobis distance. The inequality $(ii)$ is due to applying triangle inequality on the numerator. The $(iii)$ equality is by replacing the denominator with the above assumption. The $(iv)$ inequality is by applying the supporting Lemma \ref{lemma:lwdksupport}. By recursively applying Lemma \ref{lemma:lwdksupport} we get the inequality $(v)$ which is independent of the random choices made by \oneshot. The inequality $(vi)$ is by using the lower bound on the denominator. The inequality $(vii)$ an upper bound on the second and the third term. In the final inequality is due to the same reason as for the upper barrier upper bound. Further the expected upper bound is independent of $k$. So we have,
 \begin{equation*}
  \~E_{\pi_{i-1}}\bigg[\frac{f_{\*X}(\*a_{i})}{(1+\epsilon)f_{\*X}(\*A_{i-1}) - f_{\*X}(\*C_{i-1}) + f_{\varphi_{i}}(\*A_{i})}\bigg] \leq \frac{2f_{\varphi_{i}}^{\*M_{i}}(\*a_{i})}{\mu_{i}\epsilon \sum_{j \leq i}f_{\varphi_{j}}^{\*M_{j}}(\*a_{j})} + \frac{12}{\mu_{i}\epsilon(i-1)}
 \end{equation*}
\end{proof}
In this lemma it is worth noting that the expected upper bounds on $l_{i}^{u}$ and $l_{i}^{l}$ do not use any bicriteria approximation, nor are dependent on $k$. 
In order to have these upper bounds valid for any $\*X$ with at most $k$ centres in $\~R^{d}$, where $k \leq n$, we take a union bound over all $k \in [n]$. By lemma \ref{lemma:lwdkmGuarantee} and lemma \ref{lemma:lwdkScore} we claim that our coresets are non-parametric in nature. Next we bound the expected sample size.


\begin{lemma}{\label{lemma:lwdkComplexity}}
  For the above setup the term $\sum_{i \leq n} c^{u}l_{i}^{u} + c^{l}l_{i}^{l}$ is $O\Big(\frac{1}{\mu\epsilon^{2}}\Big(\log n + \log \big(f_{\varphi}^{\*M}(\*A)\big) - \log \big(f_{\varphi_{2}}^{\*M_{2}}(\*a_{2})\big)\Big)\Big)$.
\end{lemma}
\begin{proof}{\label{proof:lwdkComplexity}}
 Here in order to bound the expected sample size we first bound the expected sample size of coreset the coreset $\forall \*X \in \~R^{k \times d}$. For this we bound the expected sampling probability i.e., $\~E_{\pi_{i-1}}[p_{i}]$. 
 \begin{eqnarray*}
  \~E_{\pi_{i-1}}[p_{i}] &\stackrel{(i)}{=}& c^{u}l_{i}^{u} + c^{l}l_{i}^{l} \\
  &\leq& \frac{2c^{u}f_{\varphi_{i}}^{\*M_{i}}(\*a_{i})}{\mu_{i}\epsilon \sum_{j \leq i}f_{\varphi_{j}}^{\*M_{j}}(\*a_{j})} + \frac{12c^{u}}{\mu_{i}\epsilon(i-1)} + \frac{2c^{l}f_{\varphi_{i}}^{\*M_{i}}(\*a_{i})}{\epsilon \sum_{j \leq i}f_{\varphi_{j}}^{\*M_{j}}(\*a_{j})} + \frac{12c^{l}}{\mu_{i}\epsilon(i-1)} \\
  &\leq& \frac{8f_{\varphi_{i}}^{\*M_{i}}(\*a_{i})}{\mu_{i}\epsilon^{2} \sum_{j \leq i}f_{\varphi_{j}}^{\*M_{j}}(\*a_{j})} + \frac{48}{\mu_{i}\epsilon^{2}(i-1)} 
 \end{eqnarray*}
 Now we bound the total expected sample size,
\begin{eqnarray*}
 \sum_{2 \leq i \leq n}\~E[p_{i}] &\leq& \sum_{2 \leq i \leq n}\Bigg(\frac{8f_{\varphi_{i}}^{\*M_{i}}(\*a_{i})}{\mu_{i}\epsilon^{2} \sum_{j \leq i}f_{\varphi_{j}}^{\*M_{j}}(\*a_{j})} + \frac{48}{\mu_{i}\epsilon^{2}(i-1)}\Bigg) \\
 &\leq& \frac{48\log n}{\mu_{i}\epsilon^{2}} + \sum_{2 \leq i \leq n}\Bigg(\frac{8f_{\varphi_{i}}^{\*M_{i}}(\*a_{i})}{\mu_{i}\epsilon^{2} \sum_{j \leq i}f_{\varphi_{j}}^{\*M_{j}}(\*a_{j})}\Bigg)
\end{eqnarray*}
Let the term $\frac{f_{\varphi_{i}}^{\*M_{i}}(\*a_{i})}{\sum_{j \leq i}f_{\varphi_{j}}^{\*M_{j}}(\*a_{j})} = q_{i} \leq 1$. In the following analysis we bound summation of this term i.e., $\sum_{i \leq n} q_{i}$. For that consider the term $\sum_{j \leq i}f_{\varphi_{j}}^{\*M_{j}}(\*a_{j})$ as follows,
\begin{eqnarray*}
 \sum_{j \leq i}f_{\varphi_{j}}^{\*M_{j}}(\*a_{j}) &=& \sum_{j \leq i-1}f_{\varphi_{j}}^{\*M_{j}}(\*a_{j})\bigg(1 + \frac{f_{\varphi_{i}}^{\*M_{i}}(\*a_{i})}{\sum_{j \leq i-1} f_{\varphi_{j}}^{\*M_{j}}(\*a_{j})}\bigg) \\
 &\geq& \sum_{j \leq i-1}f_{\varphi_{j}}^{\*M_{j}}(\*a_{j})\bigg(1 + \frac{f_{\varphi_{i}}^{\*M_{i}}(\*a_{i})}{\sum_{j \leq i}f_{\varphi_{j}}^{\*M_{j}}(\*a_{j})}\bigg) \\
 &=& \sum_{j \leq i-1}f_{\varphi_{j}}^{\*M_{j}}(\*a_{j})(1 + q_{i}) \\
 &\geq& \exp(q_{i}/2)\sum_{j \leq i-1}f_{\varphi_{j}}^{\*M_{j}}(\*a_{j}) \\
 \exp(q_{i}/2) &\leq& \frac{\sum_{j \leq i}f_{\varphi_{j}}^{\*M_{j}}(\*a_{j})}{\sum_{j \leq i-1}f_{\varphi_{j}}^{\*M_{j}}(\*a_{j})}
\end{eqnarray*}
Now as we know that $\sum_{j \leq i}f_{\varphi_{j}}^{\*M_{j}}(\*a_{j}) \geq \sum_{j \leq i-1}f_{\varphi_{j}}^{\*M_{j}}(\*a_{j})$ hence following product results into a telescopic product and we get,
\begin{eqnarray*}
 \prod_{2 \leq i \leq n} \exp(q_{i}/2) &\leq& \frac{\sum_{j \leq n}f_{\varphi_{j}}^{\*M_{j}}(\*a_{j})}{f_{\varphi_{2}}^{\*M_{2}}(\*a_{2})} \\
 &\leq&\frac{f_{\varphi}^{\*M}(\*A)}{f_{\varphi_{2}}^{\*M_{2}}(\*a_{2})}
\end{eqnarray*}
Now taking $\log$ in both sides we get $\sum_{2 \leq i \leq n} q_{i} \leq 2\log \big(f_{\varphi}^{\*M}(\*A)\big) - 2\log \big(f_{\varphi_{2}}^{\*M_{2}}(\*a_{2})\big)$.
\begin{equation*}
 \sum_{i \leq n}\~E[p_{i}] \leq 1 + \frac{32}{\mu\epsilon^{2}}\Big(3\log n + \log \big(f_{\varphi}^{\*M}(\*A)\big) - \log \big(f_{\varphi_{2}}^{\*M_{2}}(\*a_{2})\big)\Big)
\end{equation*}

 Here we consider $\forall i \in [n]$ we have $\mu = \mu_{n} \leq \mu_{i}$ and $\*M = \*M_{n} \succeq \*M_{i}$ for $\*A$. 
\end{proof}

The above sum is independent of $k$ and $d$. Now to ensure a non-parametric coreset, we have to ensure this event (bounded sampling complexity) for all $\*X$ with at most $n$ centres in $\~R^{d}$. Hence by taking a union bound over $k \in [n]$, the expected sample size for a non-parametric coreset from \oneshot~which ensures a deterministic guarantee is $O\bigg(\frac{\log n}{\mu\epsilon^{2}}\Big(\log n + \log \big(f_{\varphi}^{\*M}(\*A)\big) - \log \big(f_{\varphi_{2}}^{\*M_{2}}(\*a_{2})\big)\Big)\bigg)$. Note that to get such non-parametric coreset, it is necessary to have such an oracle which upper bounds the upper and lower sensitivity scores. 


Now we present an algorithmic version of \oneshot under certain assumption. Let $\@X_{j}$ for $j \in [n]$ represents the query space where each query $\*X \in \@X_{j}$ has $j$ centres in $\~R^{d}$. Now for each point $\*a_{i}$ and $\@X_{j}$ we consider the following two random variables,
\begin{eqnarray*}
 r_{(i,j)}^{u}(\*X) &=& \frac{f_{\*X}(\*a_{i})}{(1+\epsilon)f_{\*X}(\*A_{i-1}) - f_{\*X}(\*C_{i-1}) + \epsilon f_{\varphi_{i}}(\*A_{i})} \\
 r_{(i,j)}^{l}(\*X) &=& \frac{f_{\*X}(\*a_{i})}{f_{\*X}(\*C_{i-1}) - (1-\epsilon)f_{\*X}(\*A_{i-1}) + \epsilon f_{\varphi_{i}}(\*A_{i})}
\end{eqnarray*}

Here randomness is over $\*X \in \@X_{j}$. Let both $r_{(i,j)}^{u}$ and $r_{(i,j)}^{l}$ follow the bounded CDF assumption in \cite{baykal2018data}. The bellow assumption is stated for $r_{(i,j)}^{u}$ for some appropriate query space $\@X_{j}$. We consider that a similar assumption is also true for $r_{(i,j)}^{l}$.
\begin{assumption}{\label{asm:cdf}}
 There is a pair of universal constants $K$ and $K'$ such that such that for each pair $(i,j)$ with $i \in [n]$ and $\@X_{j} \in \cup_{l}\@X_{l}$, the CDF of the random variables $r_{i}^{u}(\*X)$ for $\*X \in \@X_{j}$ denoted by $G_{(i,j)}()$ satisfies,
 \begin{equation*}
  G_{(i,j)}(x^{*}/K) \leq \exp(-1/K')
 \end{equation*}
 where $x^{*} = \min\{y \in [0,1]: G_{(i,j)}(y) = 1\}$.
\end{assumption}

We consider the above assumption is true for all pair of $(i,j)$, where $i$ corresponds to point $\*a_{i}$ and $j$ corresponds to the query space $\@X_{j}$. Now the following two lemma's are similar to lemma 6 and 7 in \cite{baykal2018data}. Here we state them for completeness. Lemma \ref{lemma:inputDistribtution} is stated for all pairs of $(i,j)$ such that $i \leq n$ and $\@X_{j} \in \cup_{l \leq n}\@X_{l}$.
\begin{lemma}{\label{lemma:inputDistribtution}}
 Let $K, K' > 0$ be universal constants and let $\@X_{j}$ be the query space as defined above with CDF $G_{(i,j)}(\cdot)$ satisfying $G_{(i,j)}(x^{*}/K) \leq \exp(-1/K')$, where $x^{*} = \min\{y \in [0,1]: G_{(i,j)}(y) = 1\}$. Let $\@Y_{j} = \{\*X_{1}, \*X_{2}, \ldots, \*X_{m}\}$ be a set of $m = |\@Y_{j}|$ i.i.d. samples each drawn from $\@X_{j}$. Let $\*X_{m+1} \sim \@X_{j}$ be an iid samples, then
 \begin{eqnarray*}
  \~P\Big(K\max_{\*X \in \@X_{j}}r_{(i,j)}^{u}(\*X) \leq r_{(i,j)}^{u}(\*X_{m+1})\Big) &\leq& \exp(-m/K') \\
  \~P\Big(K\max_{\*X \in \@X_{j}}r_{(i,j)}^{l}(\*X) \leq r_{(i,j)}^{l}(\*X_{m+1})\Big) &\leq& \exp(-m/K')
 \end{eqnarray*}
\end{lemma}

\begin{proof}{\label{proof:inputDistribtution}}
 Let $\*X_{\max} = \mbox{arg}\max_{\*X \in \@X_{j}} r_{(i,j)}^{u}(\*X)$, then
 \begin{eqnarray*}
  \~P(K\max_{\*X \in \@X_{j}}r_{(i,j)}^{u}(\*X) \leq r_{(i,j)}^{u}(\*X_{m+1})) &=& \int_{0}^{x^{*}} \~P(Kr_{(i,j)}^{u}(\*X_{\max}) \leq y/K | r_{(i,j)}^{u}(\*X_{m+1}) = y) d\~P(y) \\
  &\stackrel{i}{=}& \int_{0}^{x^{*}} \~P(Kr_{(i,j)}^{u}(\*X_{\max}) \leq y/K)^{m} \~P(y) \\
  &\leq& \int_{0}^{x^{*}} G_{(i,j)}(y/K)^{m} \~P(y) \\
  &\stackrel{ii}{\leq}& G_{(i,j)}(x^{*}/K)^{m} \int_{0}^{x^{*}} \~P(y) \\
  &=& G_{(i,j)}(x^{*}/K)^{m} \\
  &\leq& \exp(-m/K')
 \end{eqnarray*}
 Here $(i)$ is because $\{\*X_{1}, \*X_{2}, \ldots, \*X_{m}\}$ are i.i.d. from $\@X_{j}$. Further $(ii)$ is due to the assumption \ref{asm:cdf}. Similarly for $r_{i}^{l}$ is also proved.
\end{proof}

Let for all $j \leq n$, there is a finite set $\@Y_{j} \subset \@X_{j}$ such that the empirical sensitivity scores $\tilde{l}_{i}^{u} = \max_{j} \tilde{l}_{(i,j)}^{u}$ and $\tilde{l}_{i}^{l} = \max_{j} \tilde{l}_{(i,j)}^{l}$, such that $\tilde{l}_{(i,j)}^{u}$ and $\tilde{l}_{(i,j)}^{l}$ are defined as follows,
\begin{eqnarray*}
 \tilde{l}_{(i,j)}^{u} &=& \max_{\*X \in \@Y_{j}} \frac{f_{\*X}(\*a_{i})}{(1+\epsilon)f_{\*X}(\*A_{i-1}) - f_{\*X}(\*C_{i-1}) + \epsilon f_{\varphi_{i}}(\*A_{i})} \\
 \tilde{l}_{(i,j)}^{l} &=& \max_{\*X \in \@Y_{j}} \frac{f_{\*X}(\*a_{i})}{f_{\*X}(\*C_{i-1}) - (1+\epsilon)f_{\*X}(\*A_{i-1}) + \epsilon f_{\varphi_{i}}(\*A_{i})}
\end{eqnarray*}

Now in the following lemma we establish the notion that empirical sensitivity scores are good approximation to the true sensitivity scores. It is also used to decide the size of each finite set $\@Y_{j}$.
\begin{lemma}{\label{lemma:empircalSensitivity}}
 Let $\delta \in (0,1)$, consider the set $\@Y_{j} \subset \@X_{j}$ of size $|\@Y_{j}| \geq \lceil K' \log (n/\delta) \rceil$, then
 \begin{eqnarray*}
  \~P_{\*X \in \@X_{j}}(\exists i \in [n]: K\tilde{l}_{(i,j)}^{u} \leq r_{(i,j)}^{u}(\*X)) &\leq& \delta \\
  \~P_{\*X \in \@X_{j}}(\exists i \in [n]: K\tilde{l}_{(i,j)}^{l} \leq r_{(i,j)}^{l}(\*X)) &\leq& \delta
 \end{eqnarray*}
\end{lemma}

\begin{proof}{\label{proof:empircalSensitivity}}
 The proof is very simple, which mainly follows from lemma \ref{lemma:inputDistribtution}, 
 \begin{equation*}
  \~P(\@E_{i}): \~P_{\*X \in \@X_{j}}(K\max_{\*X' \in \@Y_{j}}r_{(i,j)}^{u}(\*X') \leq r_{(i,j)}^{u}(\*X)) \leq \exp(-|\@Y_{j}|/K')
 \end{equation*}
 Next, in order to ensure that there exists an $i \in [n]$ such that $\~P(\@E_{i}) \leq \delta$, we take a union bound over all $i \in [n]$ and get $|\@Y_{j}| \geq \lceil K' \log(n/\delta) \rceil$. Similarly, for the $\tilde{l}_{(i,j)}^{l}$ is also proved.
\end{proof}


Now based on above assumption we present  algorithm (\ref{alg:lwnpkm}), which returns coreset for clustering via Bregman divergence. Note that without the above assumption \ref{asm:cdf} our algorithm acts as an heuristic. 
In this algorithm instead of getting the $l_{i}^{u}$ and $l_{i}^{l}$ value from an oracle, we use $(\tilde{l}_{i}^{u}, \tilde{l}_{i}^{l})$ and the expected upper bounds as in Lemma \ref{lemma:lwdkScore}. The algorithm requires $\{\@Y_{1}, \@Y_{2}, \ldots, \@Y_{n}\}$ where each $\@Y_{j}$ has $O(\log (n/\delta))$ queries. Note although the algorithm uses the knowledge on upper bound of the cluster centres due to the query sets $\{\@Y_{1}, \@Y_{2}, \ldots, \@Y_{n}\}$, but the expected coreset size is independent of cluster centres, mainly by lemma \ref{lemma:lwdkScore} and \ref{lemma:lwdkComplexity}. In practice one might have the knowledge of upper bound of cluster centres e.g., $B$ such that $B \ll n$. 

Further as we also use expected upper bounds while taking the sampling decision, hence here we maintain $t$ coresets instead of $1$ and have an additional reweighing of $1/t$ for each sampled points. We maintain $t$ coresets in order to improve the chance that the expected upper bound from Lemma \ref{lemma:lwdkScore} actually upper bounds the true sensitivity scores, i.e., 
\begin{eqnarray*}
 \frac{1}{t}\sum_{j \leq t}\sup_{\*X \in \mathcal{X}} \frac{f_{\*X}(\*a_{i})}{(1+\epsilon)f_{\*X}(\*A_{i-1}) - f_{\*X}(\*C_{i-1}^{j}) + \epsilon f_{\varphi_{i}}(\*A_{i})} \leq \frac{2f_{\varphi_{i}}(\*a_{i})}{\epsilon \mu_{i}S} + \frac{12}{\epsilon\mu_{i}(i-1)} \\
 \frac{1}{t}\sum_{j \leq t}\sup_{\*X \in \mathcal{X}} \frac{f_{\*X}(\*a_{i})}{f_{\*X}(\*C_{i-1}^{j}) - (1-\epsilon)f_{\*X}(\*A_{i-1}) + \epsilon f_{\varphi_{i}}(\*A_{i})} \leq \frac{2f_{\varphi_{i}}(\*a_{i})}{\epsilon \mu_{i}S} + \frac{12}{\epsilon\mu_{i}(i-1)}
\end{eqnarray*}

Here an important point to note is that unlike \oneshot, the coreset from \oneshotnp only ensures the guarantee with high probability. This is due to lemma \ref{lemma:empircalSensitivity} where the empirical sensitivity score only approximate the true sensitivity scores with some probability and also the above upper bound is only over expectation. We propose our algorithm as \oneshotnp.

\begin{algorithm}[htpb]
  \caption{\oneshotnp}{\label{alg:lwnpkm}}
  \begin{algorithmic}
  \REQUIRE Points $\*a_{i}, i = 1,\ldots n; t > 1; \epsilon \in (0,1); \@Y = \cup_{j \leq}\@Y_{j}$
  \ENSURE $\{\text{Coreset ,Weights}\}: \{(\*C^{1},\*C^{2},\ldots,\*C^{t}), (\Omega^{1},\Omega^{2},\ldots,\Omega^{t})\}$
  \STATE $c^{u} = 2/\epsilon + 1; c^{l} = 2/\epsilon - 1; \varphi_{0} = \emptyset; S = 0; \*C_{0}^{1} = \ldots = \Omega_{0}^{t} = \emptyset$
  \STATE $\lambda = \|\*a_{1}\|_{\min}; \quad \nu = \|\*a_{1}\|_{\max}$
  \WHILE {$i \leq n$}
    \STATE $\lambda = \min\{\lambda,\|\*a_{i}\|_{\min}\}; \nu = \max\{\nu,\|\*a_{i}\|_{\max}\}$
    \STATE Update $\*M_{i};\ \mu_{i} = \lambda/\nu$ 
    \STATE $\varphi_{i} = ((i-1)\varphi_{i-1} + \*a_{i})/i ; S = S + f_{\varphi_{i}}^{\*M_{i}}(\*a_{i})$
    \IF{$i = 1$}
      \STATE $p_{i} = 1$
    \ELSE
      \STATE $\hat{l}_{i}^{u} = \frac{2f_{\varphi_{i}}^{\*M_{i}}(\*a_{i})}{\epsilon \mu_{i}S} + \frac{12}{\epsilon\mu_{i}(i-1)}$
      \STATE $\hat{l}_{i}^{l} = \frac{2f_{\varphi_{i}}^{\*M_{i}}(\*a_{i})}{\epsilon \mu_{i}S} + \frac{12}{\epsilon\mu_{i}(i-1)}$
      \STATE $\tilde{l}_{i}^{u} = \max_{\*X \in \@Y} \frac{f_{\*X}^{\*M_{i}}(\*a_{i})}{\mu_{i}((1+\epsilon)f_{\*X}^{\*M_{i}}(\*A_{i-1}) - \sum_{k \leq t}f_{\*X}^{\*M_{i}}(\*C_{i-1}^{k}) + \epsilon f_{\varphi_{i}}^{\*M_{i}}(\*A_{i}))}$
      \STATE $\tilde{l}_{i}^{l} = \max_{\*X \in \@Y} \frac{f_{\*X}^{\*M_{i}}(\*a_{i})}{\mu_{i}(\sum_{k \leq t}f_{\*X}^{\*M_{i}}(\*C_{i-1}^{k}) - (1+\epsilon)f_{\*X}^{\*M_{i}}(\*A_{i-1}) + \epsilon f_{\varphi_{i}}^{\*M_{i}}(\*A_{i}))}$ 
      \STATE $p_{i} = \min\{1,(c^{u}(\hat{l}_{i}^{u}+\tilde{l}_{i}^{u}) + c^{l}(\hat{l}_{i}^{l}+\tilde{l}_{i}^{l})\}$
    \ENDIF
    \STATE Set $\*c_{i}^{1}$ and $\omega_{i}^{1}$ as \\
    $\begin{cases}\*a_{i} \mbox{ and } 1/(tp_{i}) \qquad \quad \mbox{w. p. } p_{i} \\  
    \emptyset \mbox{ and } 0 \qquad \qquad \qquad \mbox{else} \end{cases}$
    \STATE $(\*C_{i}^{1}, \Omega_{i}^{1}) = (\*C_{i-1}^{1}, \Omega_{i-1}^{1}) \cup (\*c_{i}^{1},\omega_{i}^{1})$
    \STATE Set $\*c_{i}^{2}$ and $\omega_{i}^{2}$ as \\
    $\begin{cases}\*a_{i} \mbox{ and } 1/(tp_{i}) \qquad \quad \mbox{w. p. } p_{i} \\  
    \emptyset \mbox{ and } 0 \qquad \qquad \qquad \mbox{else} \end{cases}$
    \STATE $(\*C_{i}^{2}, \Omega_{i}^{2}) = (\*C_{i-1}^{2}, \Omega_{i-1}^{2}) \cup (\*c_{i}^{2},\omega_{i}^{2})$
    \STATE $\vdots$
    \STATE Set $\*c_{i}^{t}$ and $\omega_{i}^{t}$ as \\
    $\begin{cases}\*a_{i} \mbox{ and } 1/(tp_{i}) \qquad \quad \mbox{w. p. } p_{i} \\  
    \emptyset \mbox{ and } 0 \qquad \qquad \qquad \mbox{else} \end{cases}$
    \STATE $(\*C_{i}^{t}, \Omega_{i}^{t}) = (\*C_{i-1}^{t}, \Omega_{i-1}^{t}) \cup (\*c_{i}^{t},\omega_{i}^{t})$
  \ENDWHILE
  \STATE Return $\{(\*C^{1}, \*C^{2}, \ldots \*C^{t}), (\Omega^{1}, \Omega^{2}, \ldots, \Omega^{t})\}$
  \end{algorithmic}
\end{algorithm}

Note that \oneshotnp is computationally expensive for computing the terms $\tilde{l}_{i}^{u}$ and $\tilde{l}_{i}^{l}$. Now if the assumption \ref{asm:cdf} on $G_{(i,j)}$ is not true then the algorithm becomes a heuristic. The possible query spaces $\@X_{j}$ can have a query $\*X$ which has centres from the set of input points, or centres chosen randomly from $\{\mu-R,\mu+R\}^{d}$ where $\mu$ is the mean of input points and the farthest point from $\mu$ is at a distance $R$. We discuss these in detail in our revised version, where we also present appropriate empirical results for non-parametric coreset\footnote{Work under progress}.
\subsection{Coresets for k-means Clustering}
Again in the case of k-means clustering the algorithm \oneshot shows an existential coreset $\*C$ which is non-parametric in nature. In the following corollary we state guarantee that \oneshot ensures in the case of k-means clustering.
\begin{corollary}{\label{cor:nplwk}}
 Let $\*A \in \~R^{n \times d}$ and the points fed to \oneshot, it returns a set of coresets $\*C$ which ensures the guarantee as in equation \eqref{eq:lwkm} for any $\*X$ with at most $n$ centres in $\~R^{d}$. The returned coresets has expected samples size as $O\bigg(\frac{\log n}{\epsilon^{2}}\Big(\log n + \log \big(f_{\varphi}(\*A)\big) - \log \big(f_{\varphi_{2}}(\*a_{2})\big)\Big)\bigg)$. 
\end{corollary}
\begin{proof}{\label{proof:lwdkbd}}
 We prove it using the Lemmas \ref{lemma:lwdksupport}, \ref{lemma:lwdkScore} and \ref{lemma:lwdkComplexity}. As for k-means clustering we have $\*M_{i} = \*I_{d}$ and $\mu_{i} = 1$ for each $i \leq n$, hence $\forall i \in [n]$ we have,
 \begin{equation*}
  l_{i} = \frac{f_{\varphi_{i}}(\*a_{i})}{\epsilon \sum_{j \leq i}f_{\varphi_{j}}(\*a_{j})} + \frac{12}{\epsilon(i-1)}
 \end{equation*} 
 It can be verified by a similar analysis as in the proof of Lemma \ref{lemma:lwdksupport} and \ref{lemma:lwdkScore}. The rest of the lemma's proof follows as it is and we get a required guarantee, i.e., a non-parametric coreset for k-means clustering with $k \leq n$. The coreset returned by \oneshot has expected sample size as $O\bigg(\frac{\log n}{\epsilon^{2}}\Big(\log n + \log \big(f_{\varphi}(\*A)\big) - \log \big(f_{\varphi_{2}}(\*a_{2})\big)\Big)\bigg)$.
\end{proof}
As this is just an existential result, we present a heuristic algorithm for the same problem in section \ref{sec:npHeuristic}. Further we show that the coreset returned by our heuristic algorithm captures the non-parametric nature of the coreset and performs well on real world data.
\subsubsection{Coresets for DP-Means Clustering}
Here we discuss that our existential non-parametric coresets from algorithm \ref{alg:lwdkm} can also be used to approximate DP-Means clustering \cite{bachem2015coresets}, based on squared euclidean. We define a slightly different cost. 
\begin{equation*}
  cost_{DP}(\*A,\*X) = \sum_{\*a_{i} \in \*A}f_{\*X}(\*a_{i}) + |\*X|\lambda
\end{equation*}

Here $f_{\*X}(\*a_{i})$ is the cost based on some $d_{\Phi}$ Bregman divergences as introduced earlier. It is not difficult to see that the coreset from \oneshot ensures an additive error approximation for this definition of Bregman divergence based DP-Means clustering. 
\begin{lemma}{\label{lemma:DP-Means}}
 The non-parametric coreset $\*C$ from \oneshot~ensures the following for all $\*X$ with at most $n$ centres in $\~R^{d}$,
 \begin{equation*}
  |cost_{DP}(\*C,\*X) - cost_{DP}(\*A,\*X)| \leq \epsilon(f_{\*X}(\*A) + f_{\varphi}(\*A))
 \end{equation*}
\end{lemma}
\begin{proof}
Note that at any $|\*X| \ge  n$, the cost is at least $\lambda|\*X| \ge n\lambda$.
Hence without loss of generality, we can restrict $|\*X| \le n$, since the optimal will be in this range. We know that for a parameter $\lambda$,
 \begin{equation*}
  cost_{DP}(\*A,\*X) = f_{\*X}(\*A) + |\*X|\lambda
 \end{equation*}
 Now if one applies DP-Means on the coreset from \oneshot we get the following,
 \begin{eqnarray*}
  \Big|\sum_{j \leq t}\mbox{cost}_{DP}(\*C,\*X) - \mbox{cost}_{DP}(\*A,\*X)\Big| &=& \Big|f_{\*X}(\*C) - f_{\*X}(\*A)\Big| \\
  &\leq& \epsilon(f_{\*X}(\*A) + f_{\varphi}(\*A))
 \end{eqnarray*}
 The last inequality is by Theorem \ref{thm:lwdkm}.
\end{proof}
Now we claim that by allowing a small additive error approximation our coreset size significantly improves upon coresets for relative error approximation for DP-Means clustering. Unlike \cite{bachem2015coresets} our coresets are existential but it is much smaller, as in practice $O(d^d) \gg O(\log n)$.

\begin{theorem}{\label{thm:DPMeans}}
 For $\epsilon \in (0,1)$, let $\*C$ be the existential non-parametric coreset for $\*A$ (Theorem \ref{thm:lwdkm}), $\*X_{\*C}$ and $\*X_{\*A}$ are the optimal cluster centers for the DP-Means clustering on $\*C$ and $\*A$. Then $\*C$ ensures the following,
\begin{equation*}
  \mbox{cost}_{DP}(\*A,\*X_{\*C}) \leq \mbox{cost}_{DP}(\*A,\*X_{\*A}) + \epsilon(f_{\*X_{\*C}}(\*A) + f_{\*X_{\*A}}(\*A) + 2 f_{\varphi}(\*A))
\end{equation*}
The expected size of such existential coreset $\*C$ is $O\Big(\frac{\log n}{\mu\epsilon^{2}}\Big(\log n + \log \big(f_{\varphi}^{\*M}(\*A)\big) - \log \big(f_{\varphi_{2}}^{\*M_{2}}(\*a_{2})\big)\Big)\Big)$.
\end{theorem}
\begin{proof}{\label{proof:DPMeans}}
 Let $\*X_{\*C}$ and $\*X_{\*A}$ are the optimal cluster centres for DP-Means clustering on $\*C$ and $\*A$ respectively. Now we know that,
 \begin{eqnarray*}
  cost_{DP}(\*A,\*X_{\*C}) - \epsilon(f_{\*X_{\*C}}(\*A) + f_{\varphi}(\*A)) &\leq& \mbox{cost}_{DP}(\*C,\*X_{\*C}) \\
  &\leq& \mbox{cost}_{DP}(\*C,\*X_{\*A}) \\
  &\leq& cost_{DP}(\*A,\*X_{\*A}) + \epsilon(f_{\*X_{\*A}}(\*A) + f_{\varphi}(\*A)) \\
  cost_{DP}(\*A,\*X_{\*C}) &\leq& cost_{DP}(\*A,\*X_{\*A}) + \epsilon(f_{\*X_{\*C}}(\*A) + f_{\*X_{\*A}}(\*A) + 2 f_{\varphi}(\*A))
 \end{eqnarray*}
 Here the first inequality is due to lemma \ref{lemma:DP-Means}. The rest of inequalities are due to strong coreset guarantee of $\*C$.
\end{proof}

%% file: sections/experiment.tex
\section{Experiments}
In this section we demonstrate the performance of our algorithm \okm. We also demonstrate the heuristic algorithm \oneshotnp for which we empirically show that the returned coreset captures the non-parametric nature.  
\subsection{Experiments: \okm}
Here we empirically show that the coresets constructed using our proposed online algorithms outperform the baseline coreset construction algorithms.
We compare the performance of our algorithm with other baselines such as \uni~ and \tp  in solving the clustering problem.
We consider that each of the following described algorithm receives data in a streaming fashion.
\par{1) {\tt{ParaFilter}}:} Our \okm Algorithm \ref{alg:lwkm}.
\par{2) {\tp}:} This is similar to ours \okm algorithm, but has the knowledge of $\varphi$. 
i.e. substituting $\varphi_{i} = \varphi$, $\forall i$, in Algorithm \ref{alg:lwkm}.
\par{3) {\uni}:} Each arrived point here is sampled with probability $r/n$, where $r$ is a parameter used to control the expected number of samples. 

%
%
We compare the performance of the above described algorithms on following datasets: 
    \par{1) {\tt KDD(BIO-TRAIN)}}: $145,751$ samples with $74$ features.
    \par {2) {\tt SONGS}}: $515,345$ songs from the Million song dataset with $90$ features.
    
    For these datasets we consider $k=100$ and $k=200$ and consider squared Euclidean as Bregman divergence (see Figure \ref{fig:kmeans}).
    \par {3) {\tt MNIST}}: $60,000$ $28 \times 28$ dimension digits dataset. We consider $k=5$ and $k=10$ on this dataset and relative entropy as Bregman divergence (see Figure \ref{fig:kl}).

\begin{figure}[htbp]
    \centering
    \centerline{\includegraphics[scale=0.55]{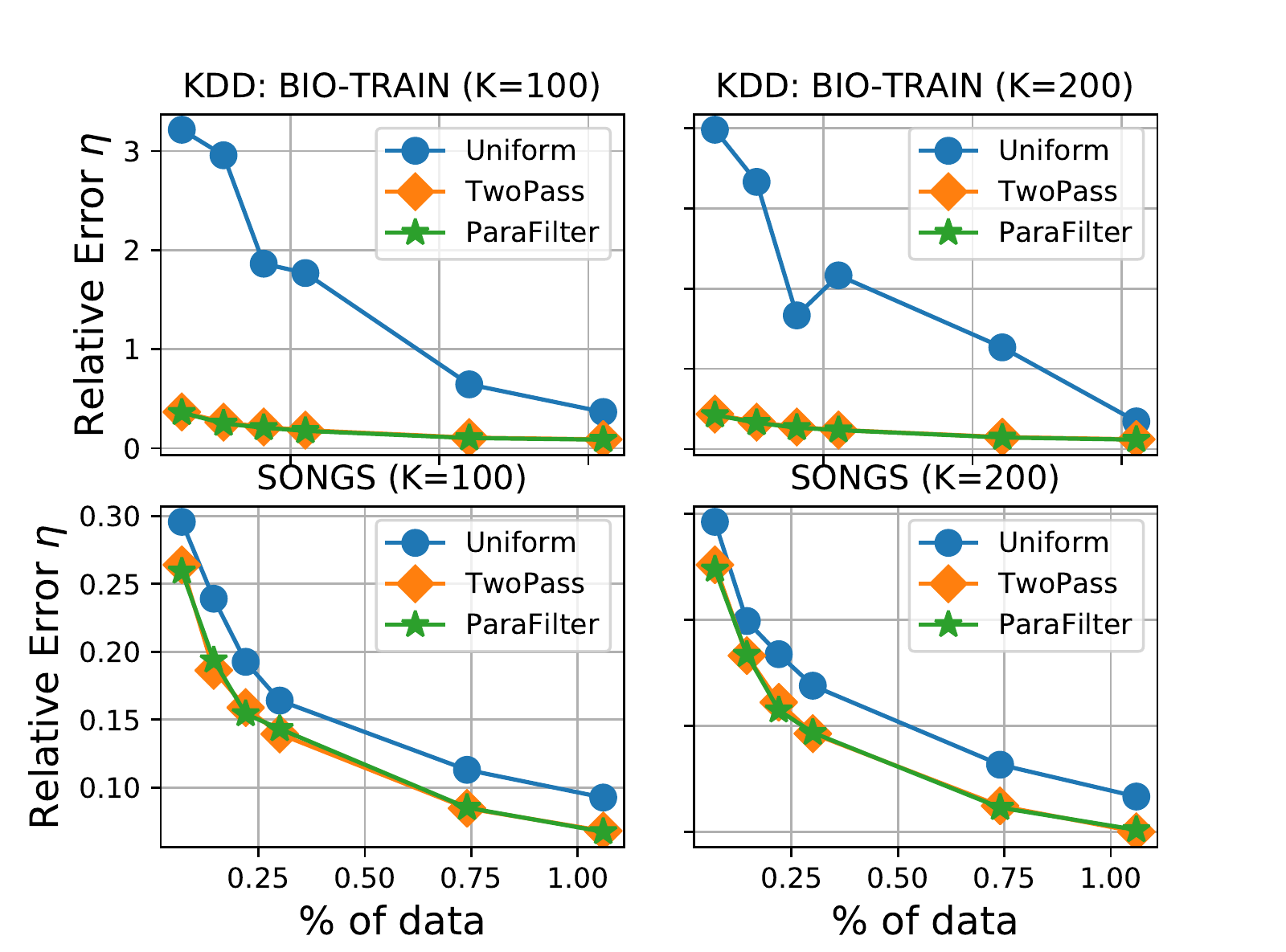}}
    \vspace{-3pt}
    \caption{Relative error v/s coreset size. Squared Euclidean Distance as Bregman Divergence.}
    \label{fig:kmeans}
\end{figure}
\begin{figure}[htbp]
    \centering
    \centerline{\includegraphics[scale=0.47, keepaspectratio]{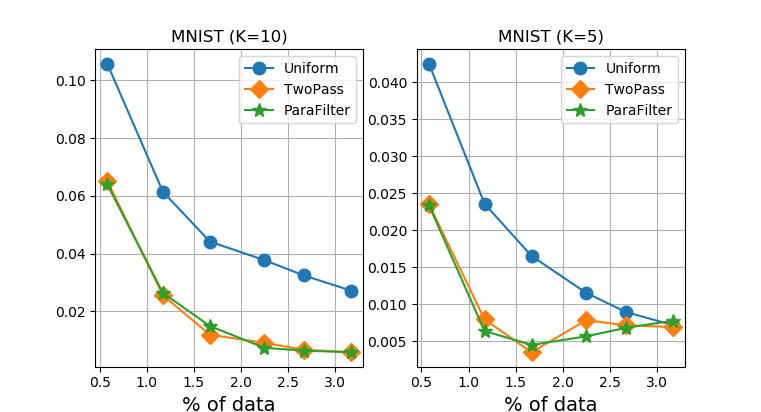}}
    \vspace{-3pt}
    \caption{Relative error v/s coreset size. Relative Entropy as Bregman Divergence.}
    \label{fig:kl}
\end{figure}

Using each of the above described algorithm, we subsample coresets of different sizes.
Once the coreset is generated, we run the weighted k-means++ \cite{arthur2007k} on the coreset to obtain the centers.
We then use these centers and compute the quantization error ($C_{s}$) on the full data set.
Additionally, we compute quantization error by running k-means++ on the full data set ($C_{F}$).
And then we compare the algorithms based on the {\it Relative-Error} ($\eta$) defined as $|C_{s}-C_{F}|/C_{F}$.
The relative error mentioned in the Figures \ref{fig:kmeans} and \ref{fig:kl} is averaged over 10 runs of each of the algorithms. The sample size are in expectation.


Figure \ref{fig:kmeans} highlights the change in $\eta$ with the increase in the coreset size for $K=100, ~200$ on {\tt KDD:BIO-TRAIN} and {\tt SONGS} datasets when squared Euclidean distance is considered as the Bregman divergence.
As the coreset size increases the relative error decreases for all the algorithms.
However, our algorithm \okm ({\tt ParaFilter}) outperforms \uni, and performs equivalent to that of \tp across all the datasets.
Additionally, we also compared the performance of our algorithm with the {\it Lighweight-Coreset} construction algorithm \cite{bachem2018scalable}, which is an offline algorithm.
{\it Lightweight-Coreset} algorithm performs better than that of {\tt ParaFilter}, but the difference is tiny.

Similarly, Figure \ref{fig:kl} shows the performance of the algorithms when Relative entropy is used as Bregman divergence, on the {\tt MNIST} dataset.Our algorithm {\tt ParaFilter} outperforms \uni  and performs equivalent to that of {\tp}.

\subsection{Experiments: \oneshotnp}{\label{sec:npHeuristic}}
Here we empirically show that the coreset returned by the heuristic algorithm captures the non-parametric nature. In this algorithm instead of getting the $l_{i}^{u}$ and $l_{i}^{l}$ value from an oracle, we use the expected upper bounds as in Lemma \ref{lemma:lwdkScore}. Note that without any assumption on query space \ref{asm:cdf} and without using empirical sensitivity scores \oneshotnp is a heuristic method where for each point $\*a_{i}$ it only needs to update $\varphi_{i}, \mu_{i}$ and $\mu_{i}$ to decide the sampling probability. In this case \oneshotnp~is an online algorithm, which takes decision about sampling a point $\*a_{i}$ before processing $\*a_{i+1}$. 
\subsubsection{Evaluated Algorithms}
We run our algorithm heuristic for Bregeman divergence as euclidean distance i.e., k-means clustering problem. We run \oneshotnp along with other baseline coreset creation algorithms such as \uni, \off~and \tp, create coresets for various $\epsilon$ values, such as $(1.0, 0.75, 0.50, 0.25)$ and compare their performances.
Following is the brief summary of the algorithms used for coreset creation:
\begin{enumerate}
    \item{\uni:} The algorithm has the knowledge of $n$ and it samples each point with probability $r/n$, where $r$ controls the expected number of samples.
    \item{\off:} We run offline version of lightweight coresets \cite{bachem2018scalable}.
    \item{\tp} (Filter-2-Pass): The algorithm has the knowledge of the $\varphi$, i.e., the mean point of $\*A$. The algorithm runs the \oneshotnp, where instead of $\varphi_{i}$ it uses the $\varphi$.
    \item{\tt Filter-NP}: We run \oneshotnp algorithm \ref{alg:lwdkm} -- our proposed non-parametric algorithm. 
\end{enumerate}

To ensure that the number of points sampled by each algorithm is equivalent for a fixed value of $\epsilon$, we first run our \oneshotnp~ algorithm, and then use the number of points sampled by \oneshotnp~ as the expected number of points to be sampled by other algorithms. 

\begin{figure}[ht]
\vspace{.3in}
\centerline{\includegraphics[scale=0.5]{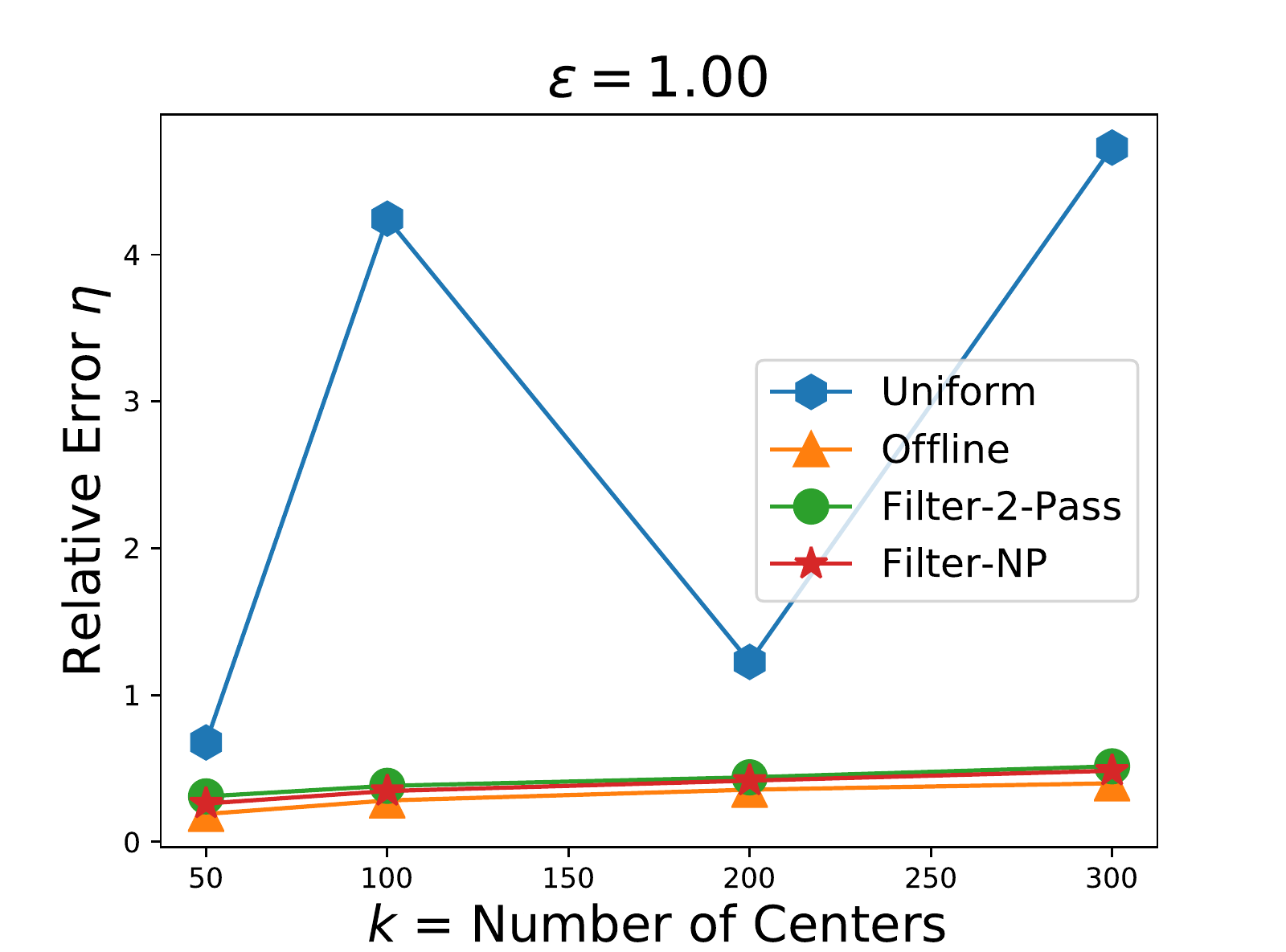} \quad \includegraphics[scale=0.5]{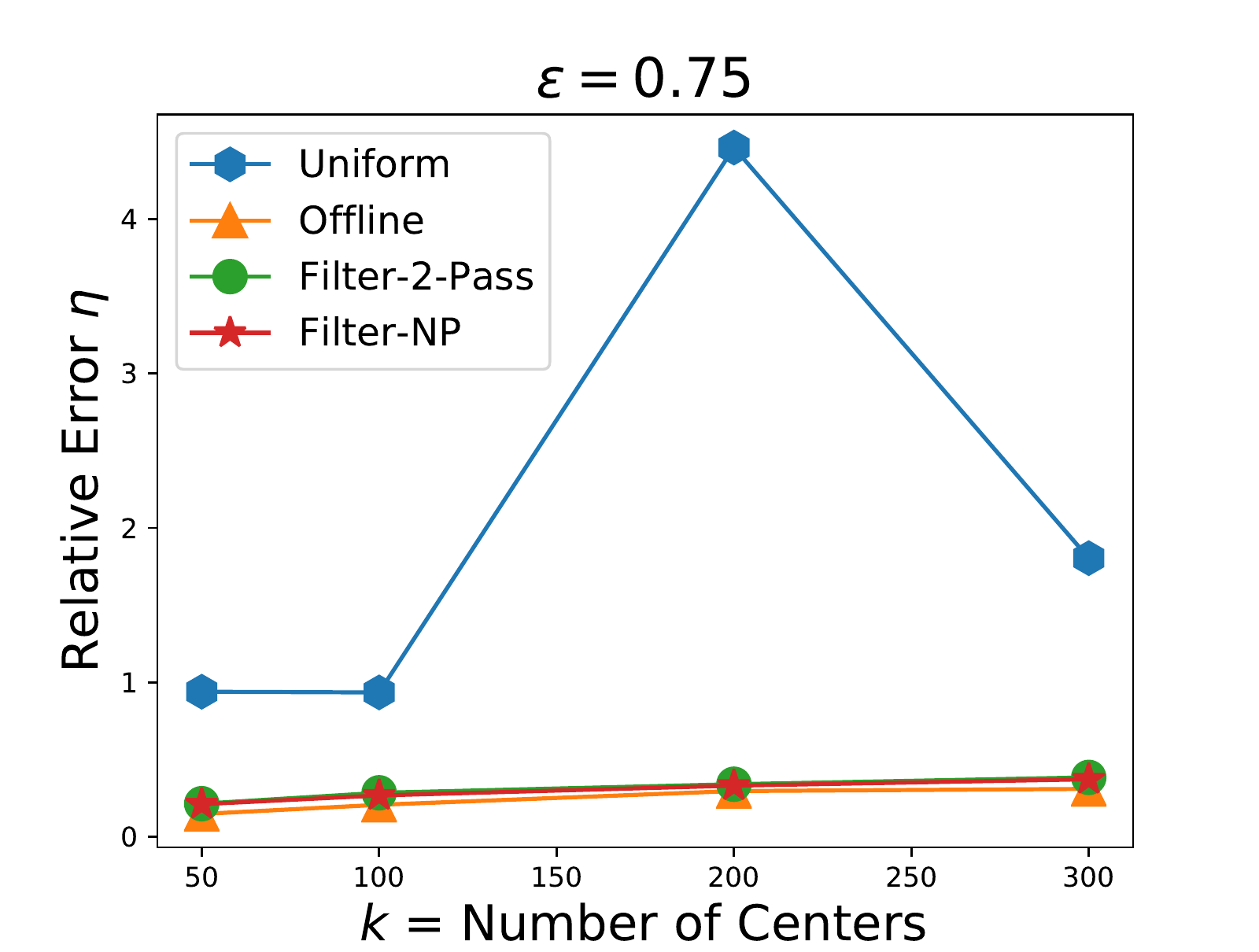}}
\centerline{\includegraphics[scale=0.5]{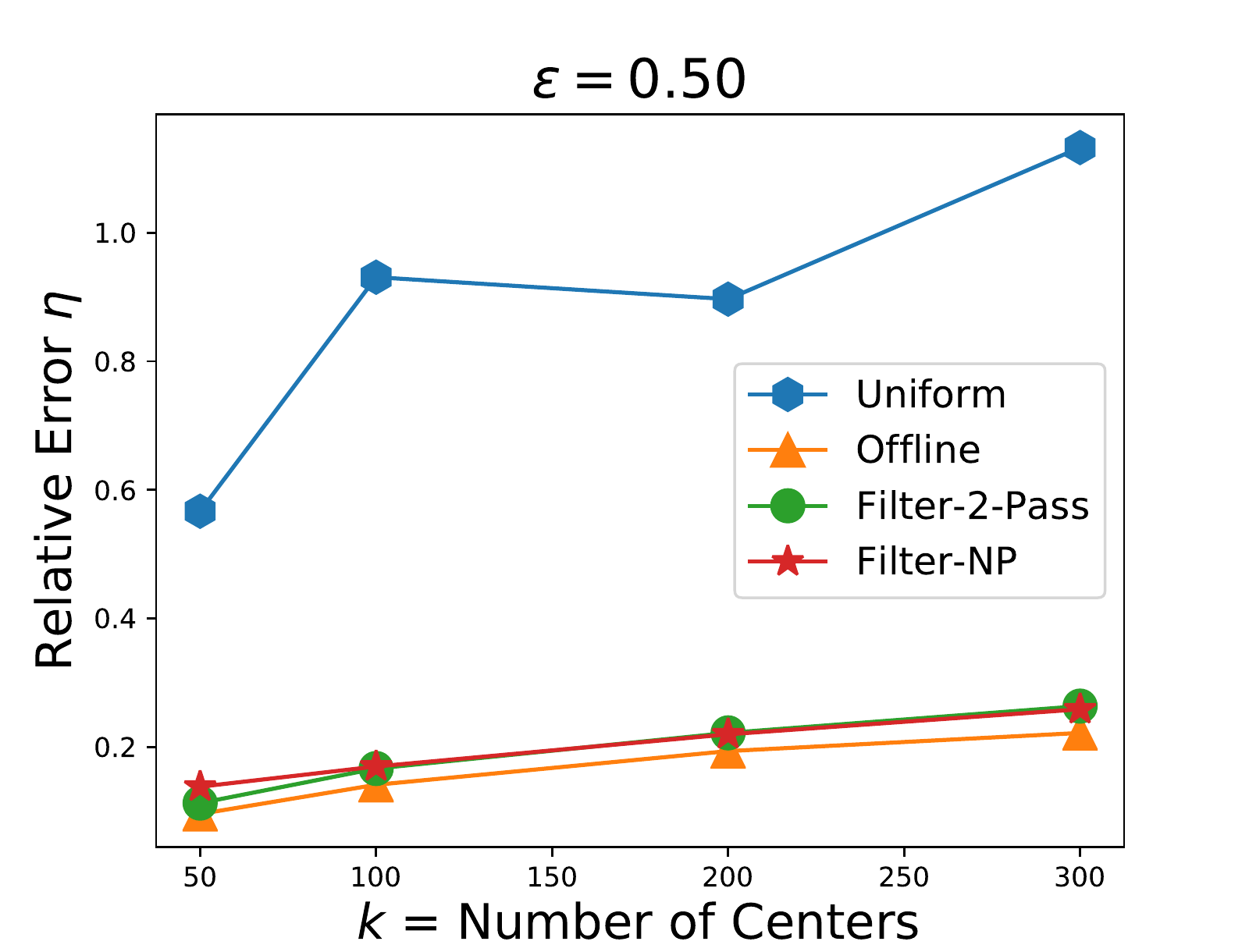} \quad \includegraphics[scale=0.5]{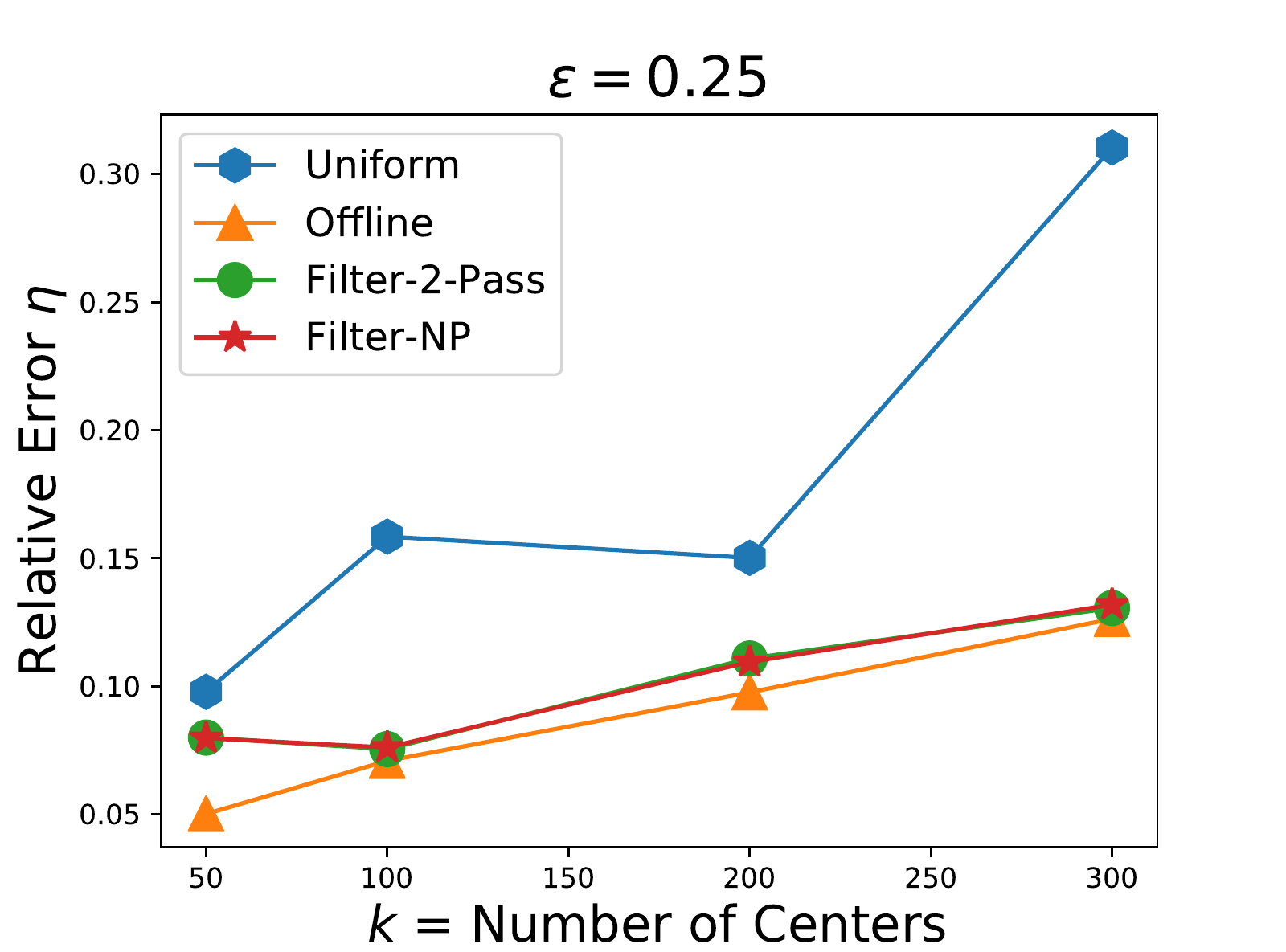}}
\vspace{.3in}
\caption{Change in Relative Error $\eta$ with respect to number of centers $k$ for various values of $\epsilon$.}
\label{fig:np-coreset-kmeans}
\end{figure}

\subsubsection{Evaluation Metric}
We evaluate the performance of the above mentioned algorithms on the {\tt KDD(BIO-TRAIN)} dataset which has $145, 751$ samples with $74$ features. 
Once the coreset is obtained from each of the sampling methods, we run weighted k-means++ clustering \cite{arthur2007k} on them for various values of $k$ such as $(50,100,200,300)$ and get the centers. These centers are considered as initial centres while running k-means clustering on the coreset and finally obtain the centres.
Once the centers are obtained, we compute the quantization error on the entire dataset with respect to the corresponding centers, i.e. $C_{S}(\*A)$ where $S$ is the set of centers returned from the coreset. 
We also run $k$-means clustering on the entire data for these values of $k$ get the quantization error, i.e., $C_F(\*A)$ where $F$ is the set of centers obtained by running $k$-means++ on the entire set of points $\*A$. 
Finally we report the relative error $\eta$, i.e., $\eta = \frac{|C_S(\*A)-C_F(\*A)|}{C_F(\*A)}$. 

For each of the algorithms mentioned above, for each value of $\epsilon$ we run $5$ random instances, compute $\eta = \frac{|C_S(\*A)-C_F(\*A)|}{C_F(\*A)}$ for each of the instances and report the median of the $\eta$ values.
We consider various values of $\epsilon$ such as $\{1.0, 0.75, 0.5, 0.25\}$, and for each value of $\epsilon$ the approximate number of points sampled are $\{500, 850, 1650, 5500\}$ respectively. Note that to capture the notion of the non-parametric nature in the coreset, we run k-means clustering for a fixed coreset for different value of $k$, i.e., $\{50, 100, 200, 300\}$. 

\subsubsection{Results}
Figure \ref{fig:np-coreset-kmeans} shows the change in the value of Relative Error $\eta$ with respect to the change in number of centers $k$, for various values of $\epsilon$.
With the decrease in the value of $\epsilon$ we can note that the value of $\eta$ overall decreases for all the algorithms.
This is as expected, because with the decrease in $\epsilon$, the additive error part decreases and the coreset size increases.
We can also observe that for each value of  $\epsilon$, with the change in the value of $k$, the relative error $\eta$ remains almost constant for all the algorithms except \uni, which shows the non-parametric nature of the importance sampling algorithms.
Also, we can note that, even if \off beats our algorithm \oneshotnp~in terms of the  
relative error, the difference is small. The above empirical results provide an evidence that such a coreset can also be used to learn extreme clustering \cite{kobren2017hierarchical}. 

%% file: sections/ack.tex
\section{Conclusion}
In this work we present online algorithm \okm that returns a coreset for clustering based on Bregman divergences. With \oneshot we further show that non-parametric coresets with additive error exist, but finding such a coreset in practice in an efficient manner is yet an open question. We also show that the existential coreset can also be used to get a well approximate solution for DP-Means problem. Without any assumption, the existential coreset relies on an oracle whose implementation is not known. The work related to \oneshotnp is under progress. In our next revision we will plan to discuss the assumption \ref{asm:cdf} in greater details. Further we also plan to provide appropriate empirical results demonstrating \oneshotnp.

\section{Acknowledgements}
We are grateful to the anonymous reviewers for their helpful feedback.
This project has received funding from the Engineering and Physical Sciences Research Council, UK (EPSRC) under Grant Ref: EP/S03353X/1. 
Anirban acknowledges the kind support of the N. Rama Rao Chair Professorship at IIT Gandhinagar, the Google India AI/ML award (2020), Google Faculty Award (2015), and CISCO University Research Grant (2016).